\documentclass[12pt,draftclsnofoot,onecolumn]{IEEEtran}
\usepackage{epsfig,latexsym,amsmath,amssymb,setspace,cite,color,graphicx,algorithm,algorithmic,cases}
\usepackage{amssymb,amsmath,amsfonts,bm,epsfig,graphicx,latexsym}
\usepackage{rotating,setspace,latexsym,epsf,color}
\usepackage{cite,authblk,epstopdf,footnote}
\usepackage{float}
\usepackage[inline]{enumitem}
\usepackage{tabu}
\usepackage{array}
\usepackage[subrefformat=parens,labelformat=parens]{subfig}
\usepackage{bbm}
\usepackage{amsthm}
\usepackage{thmtools}
\usepackage{setspace}
\usepackage{multirow} 
\usepackage{algorithm,algorithmic}

\usepackage{eqparbox}

\declaretheorem[style=plain,qed=$\blacksquare$]{theorem}
\declaretheorem[style=plain,name=Definition,qed=$\blacksquare$]{Definition}

\declaretheorem[style=plain,name=Remark,qed=$\blacksquare$]{remark}
\declaretheorem[style=plain,name=Example,qed=$\blacksquare$]{example}

\allowdisplaybreaks

\def\mc{\ensuremath\mathcal}

\newcounter{opt_ct}
\stepcounter{opt_ct}

\begin{document}


\title{Device-to-Device Coded Caching with Distinct Cache Sizes\thanks{ This work was presented in part at the IEEE International Conference on Communications (ICC), Kansas City, MO, 2018. This work was supported in part by NSF grant CCF-1749665.}
}

\author{
\IEEEauthorblockN{Abdelrahman M. Ibrahim, Ahmed A. Zewail and Aylin Yener}\\
    \IEEEauthorblockA{Wireless Communications and Networking Laboratory (WCAN)\\
    School of Electrical Engineering and Computer Science\\
    The Pennsylvania State University, University Park, PA 16802\\
   \textit{ami137@psu.edu \ \ \ zewail@psu.edu \ \ \ yener@ee.psu.edu}
}
}
\maketitle
\vspace{-0.3in}
\begin{abstract}
This paper considers a cache-aided device-to-device (D2D) system where the users are equipped with cache memories of different size. During low traffic hours, a server places content in the users' cache memories, knowing that the files requested by the users during peak traffic hours will have to be delivered by D2D transmissions only. The worst-case D2D delivery load is minimized by jointly designing the uncoded cache placement and linear coded D2D delivery. Next, a novel lower bound on the D2D delivery load with uncoded placement is proposed and used in explicitly characterizing the minimum D2D delivery load (MD2DDL) with uncoded placement for several cases of interest. In particular, having characterized the MD2DDL for equal cache sizes, it is shown that the same delivery load can be achieved in the network with users of unequal cache sizes, provided that the smallest cache size is greater than a certain threshold. The MD2DDL is also characterized in the small cache size regime, the large cache size regime, and the three-user case. Comparisons of the server-based delivery load with the D2D delivery load are provided. Finally, connections and mathematical parallels between cache-aided D2D systems and coded distributed computing (CDC) systems are discussed.

\end{abstract}
\vspace{-0.05in}
\begin{IEEEkeywords}
\vspace{-0.1in}
Coded caching, uncoded placement, device-to-device communication, unequal cache sizes. 
\end{IEEEkeywords}
\vspace{-0.1in}


\newpage
\section{Introduction}

Development of novel techniques that fully utilize network resources is imperative to meet the objectives of 5G systems and beyond with increasing demand for wireless data traffic, e.g., video-on-demand services \cite{cisco}. Device-to-device (D2D) communications \cite{asadi2014survey} and caching \cite{maddah2016coding} are two prominent techniques for alleviating network congestion. D2D communications utilize the radio interface enabling the nodes to directly communicate with each other to reduce the delivery load on servers/base stations/access points. Caching schemes utilize the nodes' cache memories to shift some of the network traffic to low congestion periods. In coded caching \cite{maddah2014fundamental}, the server jointly designs the content placement during off-peak hours and the content delivery during peak hours, to create multicast coding opportunities. That is, coded caching not only shifts the network traffic to off-peak hours but also creates multicast opportunities that reduce the delivery load on the server \cite{maddah2014fundamental}. In particular, in the placement phase, the server first partitions the files into pieces. Then, the server either places uncoded or coded pieces of the files at the users' cache memories. Most of the work on coded caching considers uncoded placement \cite{maddah2014fundamental,maddah2015decentralized,ji2016fundamental,yu2018exact,
wan2016optimality,wan2017novel,wang2015fundamental,amiri2017decentralized,sengupta2016layered,
ibrahim2017centralized,ibrahim2017optimization,ibrahim2018coded}, for its practicality and near optimality \cite{wan2016optimality,wan2017novel,yu2018exact}. References \cite{wan2016optimality,wan2017novel} have illustrated that the server-based delivery problem in \cite{maddah2014fundamental} is equivalent to an index-coding problem and the delivery load in \cite{maddah2014fundamental} is lower bounded by the acyclic index-coding bound \cite[Corollary 1]{arbabjolfaei2013capacity}. Reference \cite{yu2018exact} has proposed an alternative proof for the uncoded placement bound \cite{wan2016optimality,wan2017novel} using a genie-aided approach.


Coded caching in device-to-device networks has been investigated in \cite{ji2016fundamental,ji2016wireless,ji2017fundamental,shabani2016mobility,
tebbi2017coded,chorppath2017network,awan2015fundamental,zewail2018device}. In particular, D2D coded caching was first considered in \cite{ji2016fundamental}, where centralized and decentralized caching schemes have been proposed for when the users have equal cache sizes. References \cite{ji2016fundamental,ji2016wireless,ji2017fundamental,shabani2016mobility} have studied the impact of coded caching on throughput scaling laws of D2D networks under the protocol model in \cite{gupta2000capacity}. Reference \cite{tebbi2017coded} has considered a D2D system where only a subset of the users participate in delivering the missing subfiles to all users. Reference \cite{chorppath2017network} has proposed using random linear network coding to reduce the delay experienced by the users in lossy networks. Reference \cite{awan2015fundamental} has proposed a secure D2D delivery scheme that protects the D2D transmissions in the presence of an eavesdropper. Reference \cite{zewail2018device} has considered secure D2D coded caching when each user can recover its requested file and is simultaneously prevented from accessing any other file.

More realistic caching models that reflect the heterogeneity in content delivery networks consider systems with distinct cache sizes \cite{yang2018coded,wang2015fundamental,amiri2017decentralized,sengupta2016layered,
ibrahim2017centralized,ibrahim2017optimization,ibrahim2018coded,daniel2017optimization,
cao2018coded}, unequal file sizes \cite{zhang2015codedf,li2017rate,daniel2017optimization}, distinct distortion requirements \cite{yang2018coded,hassanzadeh2015distortion,ibrahim2018distortion}, and non-uniform popularity distributions \cite{niesen2014nonunif,niesen2017coded,zhang2015coded,hachem2015effect,
hachem2017coded,jin2017structural}. In this work, we focus on the distinct cache sizes, i.e., the varying storage capabilities of the users. This setup has been considered in \cite{yang2018coded,wang2015fundamental,amiri2017decentralized,sengupta2016layered,
ibrahim2017centralized,ibrahim2017optimization,ibrahim2018coded,daniel2017optimization,
cao2018coded} for the server-based delivery problem of \cite{maddah2014fundamental}. In particular, in \cite{ibrahim2017centralized,ibrahim2018coded}, we have shown that the delivery load is minimized by solving a linear program over the parameters of the uncoded placement and linear delivery schemes.

Different from \cite{ibrahim2017centralized,ibrahim2018coded} and all references with distinct cache sizes, in this paper, we investigate coded caching with end-users of unequal cache sizes when the delivery phase must be carried out by D2D transmissions. That is, the placement and delivery design must be such that the server does not participate in delivery at all, thus saving its resources to serve those outside the D2D network. This distinction calls for new placement and delivery schemes as compared to serve-based delivery architectures \cite{yang2018coded,wang2015fundamental,amiri2017decentralized,sengupta2016layered,
ibrahim2017centralized,ibrahim2017optimization,ibrahim2018coded,daniel2017optimization,
cao2018coded}. In the same spirit as \cite{ibrahim2017centralized}, we show that a linear program minimizes the D2D delivery load by optimizing over the partitioning of the files in the placement phase and the size and structure of the D2D transmissions, and find the optimal design. 

Building on the techniques in \cite{wan2016optimality,wan2017novel,yu2018exact}, we derive a lower bound on the worst-case D2D delivery load with uncoded placement, which is also defined by a linear program. Using the proposed lower bound, we first prove the optimality of the caching scheme in \cite{ji2016fundamental} assuming uncoded placement for systems with equal cache sizes. Next, we explicitly characterize the D2D delivery load memory trade-off assuming uncoded placement for several cases of interest. In particular, we show that the D2D delivery load depends only on the total cache size in the network whenever the smallest cache size is greater than a certain threshold. For a small system with three users, we identify the precise trade-off for any library size. For larger systems, we characterize the trade-off in two regimes, i.e., the small total cache size regime and in the large total cache size regime, which are defined in the sequel. For remaining sizes of the total network cache, we observe numerically that the proposed caching scheme achieves the minimum D2D delivery load assuming uncoded placement. Finally, we establish the relationship between the server-based and D2D delivery loads assuming uncoded placement. We also discuss the parallels between the recent coded distributed computing (CDC) framework \cite{li2018fundamental} and demonstrate how it relates to D2D caching systems.


The remainder of this paper is organized as follows. In Section \ref{sec_sysmod}, we describe the system model and the main assumptions. The optimization problems characterizing the upper and lower bounds on the minimum D2D delivery load are formulated in Section \ref{sec_formulation}. 
Section \ref{sec_results} summarizes our results on the minimum D2D delivery load with uncoded placement. The general caching scheme is developed in Section \ref{sec_cach}. Section \ref{sec_achiev} explains the caching schemes that achieve the D2D delivery loads presented in Section \ref{sec_results}. The optimality of uncoded placement is investigated in Section \ref{sec_lowerbound}. In Section \ref{sec_disc}, we discuss the trade-off in the general case, the connection to server-based systems, and connections to distributed computing. Section \ref{sec_concl} provides the conclusions.

\begin{figure}[t]
	\includegraphics[scale=1.2]{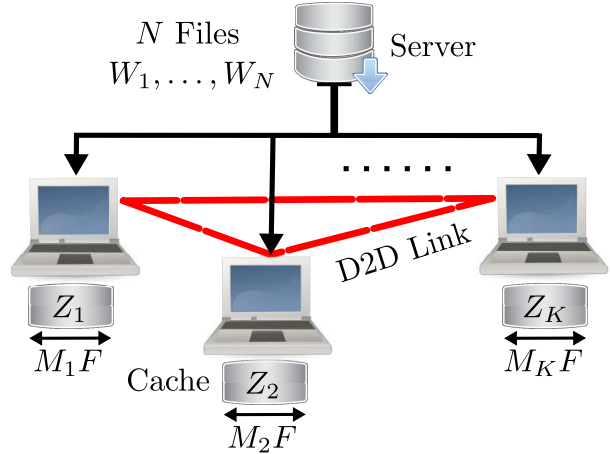}
	\centering
	\caption{D2D caching with unequal cache sizes at the end-users.}\label{fig:sys_model}
	\vspace{-.1 in}
\end{figure}


\vspace{-.1 in}
\section{System Model}\label{sec_sysmod}
\textit{Notation:} Vectors are represented by boldface letters, $ \oplus$ refers to bitwise XOR operation, $|W|$ denotes size of $W$, $\mc A \setminus \mc B $ denotes the set of elements in $\mc A$ and not in $\mc B $, $[K] \triangleq \{1,\dots,K\}$, $ \phi$ denotes the empty set, $\subsetneq_{\phi} [K]$ denotes non-empty subsets of $[K]$, and $\mc P_{\mc A}$ is the set of all permutations of the elements in the set $\mc A$, e.g., $\mc P_{\{1,2\}}= \{[1,2], \ [2,1]\}$. 

Consider a server connected to $K$ users via a shared error-free link, and the users are connected to each other via error-free device-to-device (D2D) communication links, as illustrated in Fig. \ref{fig:sys_model}. The server has a library of $N$ files, $ W_{1}, \dots, W_{N}$, each with size $F$ bits. End-users are equipped with cache memories that have different sizes, the size of the cache memory at user $k$ is equal to $M_k F$ bits. Without loss of generality, let $M_1 \leq M_2 \leq \dots \leq M_K$. Define $m_k $ to denote the memory size of user $k$ normalized by the library size $N F$, i.e., $m_k=M_k/N$. Let $\bm M =[M_1,\dots,M_K]$ and $\bm m =[m_1,\dots,m_K]$. We focus on the more practical case where the number of users is less than the number of files, i.e., $K \leq N $, e.g., a movie database serving cooperative users in a 5G hybrid cloud-fog access network \cite{ku20175g}.

D2D caching systems operate similarly to server-based systems in the placement phase, but differ in the delivery phase. Namely, in the placement phase, the server designs the users' cache contents without knowing their demands and knowing that it will not participate in the delivery phase. The content of the cache at user $k$ is denoted by $Z_k$ and satisfies the size constraint $|Z_k| \leq M_k F$ bits. Formally, $Z_k$ is defined as follows.
\begin{Definition}(Cache placement) A cache placement function $\phi_k: [2^F]^N\rightarrow [2^F]^{M_k}$ maps the files in the library to the cache memory of user $k$, i.e., $ Z_k = \phi_k(W_1, W_2,..,W_N) $. 
\end{Definition}
Just before the delivery phase, users announce their file demands; we consider that the demand vector $\bm d=[d_1, \dots, d_K]$ is independently and identically distributed (i.i.d.) over the files \cite{maddah2014fundamental}. User $k$'s requested file is denoted by $W_{d_k}$. The requested files must be delivered by utilizing D2D communications only \cite{ji2016fundamental}, which requires that the sum of the users' cache sizes is at least equal to the library size, i.e., $\sum_{k=1}^{K} m_k \geq 1$. More specifically, user $j$ transmits the sequence of unicast/multicast signals, $X_{j \rightarrow \mc T, \bm d}$, to the users in the set $\mc T \subsetneq_{\phi} [K] \! \setminus \! \{j\}$. Let $|X_{j \rightarrow \mc T, \bm d}| = v_{j \rightarrow \mc T} F$ bits, i.e., the \textit{transmission variable} $v_{j \rightarrow \mc T} \in [0,1]$ represents the amount of data delivered to the users in $\mc T$ by user $j$ as a fraction of the file size $F$. 
\begin{Definition}(Encoding) Given demand $\bm d$, an encoding $\psi_{j \rightarrow \mc T}: [2^F]^{M_j} \times [N]^{K}\rightarrow [2^{F}]^{v_{j \rightarrow \mc T}} $ maps the content cached by user $j$ to a signal sent to the users in $\mc T \subsetneq_{\phi} [K] \! \setminus \! \{j\}$, i.e., the signal $X_{j \rightarrow \mc T,\bm d}= \psi_{j \rightarrow \mc T}(Z_j,\bm d)$ and $|X_{j \rightarrow \mc T, \bm d}| = v_{j \rightarrow \mc T} F$.
\end{Definition}
At the end of the delivery phase, user $k$ must be able to reconstruct $W_{d_k}$ reliably using the received D2D signals $ \left\lbrace X_{j \rightarrow \mc T, \bm d} \right\rbrace_{j \neq k, \mc T}$ and its cache content $Z_k$. Let $R_{j} \triangleq \sum\limits_{\mc T \subsetneq_{\phi} [K] \! \setminus \! \{j\}} v_{j \rightarrow \mc T}$ be the amount of data transmitted by user $j$, normalized the file size $F$. 
\begin{Definition}(Decoding) Given the demand $\bm d$, a decoding function $\mu_{k}: [2^{F}]^{\sum_{j \neq k} R_{j}} \times [2^{F}]^{M_k} \times [N]^{K} \rightarrow [2^F]$, maps the D2D signals $X_{j \rightarrow \mc T, \bm d},$ $\forall j \in [K] \setminus \{k\}, \mc T \subsetneq_{\phi} [K] \setminus \{j\}$ and the content cached by user $k$ to $\hat W_{d_k}$, i.e., $\hat W_{d_k} = \mu_{k}\left(\left\lbrace X_{j \rightarrow \mc T, \bm d} \right\rbrace_{j \neq k, \mc T}, Z_k,\bm d \right)$.
\end{Definition}
The achievable D2D delivery load is defined as follows.
\begin{Definition} 
For a given $\bm m$, the D2D delivery load $R(\bm m)\triangleq \sum_{j=1}^{K} R_{j}(\bm m)$ is said to be achievable if for every $\epsilon > 0$ and large enough $F$, there exists $(\phi_k(.),\psi_{j \rightarrow \mc T}(.),\mu_{ k}(.))$ such that $\max\limits_{\bm d, k \in [K]} Pr(\hat W_{d_k}\neq W_{d_k})\leq \epsilon$, and $R^*(\bm m) \triangleq \inf \{ R : R(\bm m) \text{ is achievable} \}$. 
\end{Definition}
In this work, similar to much of the coded caching literature \cite{maddah2014fundamental,maddah2015decentralized,ji2016fundamental,yu2018exact,
wan2016optimality,wan2017novel,wang2015fundamental,amiri2017decentralized,sengupta2016layered,
ibrahim2017centralized,ibrahim2017optimization,ibrahim2018coded,daniel2017optimization}, we will consider placement schemes where the users cache only pieces of the files, i.e., uncoded placement. We denote the set of such schemes with $\mathfrak{A}$. In the delivery phase, we consider the class of delivery policies $\mathfrak{D}$, where the users generate the multicast signals with XORed pieces of files. For a caching scheme in $(\mathfrak{A},\mathfrak{D})$, we define the following. 
\begin{Definition}
For an uncoded placement scheme in $ \mathfrak A$, and a linear delivery policy in $\mathfrak{D}$, the achievable worst-case D2D delivery load is defined as 
\begin{equation}
R_{\mathfrak A, \mathfrak D} \triangleq \max_{\bm d \in [N]^{K}} \sum_{j=1}^{K} R_{j,\bm d,\mathfrak A, \mathfrak D}  = \sum_{j=1}^{K} \sum_{\mc T \subsetneq_{\phi} [K] \setminus \{j\}} v_{j \rightarrow \mc T},
\end{equation}
and $R^*_{\mathfrak A,\mathfrak D}$ denotes the minimum delivery load achievable with a caching scheme in $(\mathfrak A,\mathfrak D)$. 
\end{Definition}
\begin{Definition} For an uncoded placement scheme in $ \mathfrak A$ and any delivery scheme, we have
\begin{equation}
R_{\mathfrak A}^*(\bm m) \triangleq \inf \{ R_{\mathfrak A} : R_{\mathfrak A}(\bm m) \text{ is achievable} \},
\end{equation} 
which denotes the minimum D2D delivery load achievable with uncoded placement.
\end{Definition}
\section{Main Results}

\subsection{Performance Bounds}\label{sec_formulation}
First, we have the following parameterization for the optimum of the class of caching schemes under consideration. 
\begin{theorem}\label{thm_gen_achv} Given $N \geq K$, and $\bm m$, the minimum worst-case D2D delivery load assuming uncoded placement and linear delivery, $R^*_{\mathfrak{A},\mathfrak{D}}(\bm m) $, is characterized by the following linear program 
\begin{subequations} \label{eqn_opt}
\begin{align}
\textit{\textbf{O1}:}  \qquad  R^*_{\mathfrak{A},\mathfrak{D}}(\bm m) = \ &  \min_{\bm a,\bm u ,\bm v}  
& & \sum_{j=1}^{K} \sum_{\mc T \subsetneq_{\phi} [K] \setminus \{j\}} v_{j \rightarrow \mc T} \\
& \text{subject to}
& & \bm a \in \mathfrak{A}(\bm m), \\
& & & (\bm u, \bm v) \in \mathfrak D(\bm a),
\end{align}
\end{subequations}
where $\mathfrak{A}(\bm m) $ is the set of uncoded placement schemes defined in (\ref{eqn_feas_alloc}) and $\mathfrak D(\bm a) $ is the set of feasible linear delivery schemes defined by (\ref{eqn_dlv_cnst1})-(\ref{eqn_dlv_cnst5}).
\end{theorem}
\begin{proof}
Proof is provided in Section \ref{sec_cach}.
\end{proof}
%
Motivated by the lower bounds on server-based delivery in \cite{wan2016optimality,wan2017novel,yu2018exact}, we next establish that the minimum D2D delivery load memory trade-off with uncoded placement, $R^*_{\mathfrak{A}}(\bm m) $, is lower bounded by the linear program defined in Theorem \ref{thm_bound_genie}.
\begin{theorem}\label{thm_bound_genie} Given $N \geq K$, and $\bm m$, the minimum worst-case D2D delivery load with uncoded placement, $R^*_{\mathfrak{A}}(\bm m)$, is lower bounded by 
\begin{subequations} \label{eqn_bound_genie}
	\begin{align}
\textit{\textbf{O2}:}  \   & \max_{\lambda_{0} \in \mathbb{R},\lambda_{k} \geq 0,\alpha_{\bm q} \geq 0}  
	& &  -\lambda_{0} - \sum_{k=1}^{K} m_k \lambda_{k} \\
	& \text{subject to}
	& &   \! \! \! \! \! \! \! \! \! \! \! \! \lambda_{0} + \sum_{k \in \mc S} \lambda_{k} +  \gamma_{\mc S} \geq 0,  \forall \ \!\mc S  \subsetneq_{\phi} [K],\\
	& & &  \! \! \! \! \! \! \! \! \! \! \! \! \sum\limits_{\bm q \in  \mc P_{[K]\setminus\{j\}} } \alpha_{\bm q} =1, \ \! \forall \ \! j \in [K] ,  
	\end{align}
\end{subequations}
where $ \mc P_{[K]\setminus\{j\}} $ is the set of all permutations of the users in $[K] \setminus \{j\}$, and 
\begin{align}\label{eqn_gamma}
\gamma_{\mc S} \triangleq \begin{cases} K-1, \text{ for } |\mc S| =1,\\
\min\limits_{j \in \mc S} \bigg\{ \sum\limits_{i=1}^{K-|\mc S|}  \sum\limits_{\substack{\bm q \in  \mc P_{[K]\setminus\{j\}}   : \ \!  q_{i+1}  \in \mc S, \\ \{ q_1,\dots,q_{i} \} \cap \mc S=\phi } }   \! i \! \ \alpha_{\bm q} \bigg\}, \text{ for } 2 \leq \! |\mc S| \! \leq K \! - \! 1 \\
0, \text{ for } \mc S =[K].
\end{cases}
\end{align}
\end{theorem}
\vspace{-0.3in}
\begin{proof}
The proof is detailed in Section \ref{proof_thm_genie}.
\end{proof}

\subsection{Explicit Characterization Results}\label{sec_results}
Next, using Theorems \ref{thm_gen_achv} and \ref{thm_bound_genie}, we characterize the trade-off explicitly for several cases, which are illustrated in Table \ref{tab_main_results}. In particular, for these cases we show the optimality of linear delivery for uncoded placement, i.e., we show that $R^*_{\mathfrak{A}}(\bm m)= R^*_{\mathfrak{A},\mathfrak{D}}(\bm m)$. 
First, using Theorem \ref{thm_bound_genie}, we show the optimality of the centralized caching scheme proposed in \cite{ji2016fundamental} with uncoded placement for systems where the users have equal cache sizes.

\begin{theorem}\label{thm_equalcache}
For $N \geq K$, and $m_k=m=t/K, t \in [K] $, $\forall k \in [K]$, the minimum worst-case D2D delivery load with uncoded placement, $ R^*_{\mathfrak{A}}(m)= (1\! - \! m)/m$. In general, we have 
\begin{align}\label{eqn_thm_equalcache}
R^*_{\mathfrak{A}}(m)=\left(\dfrac{K-t}{t}\right)\big( t \! + \! 1 \! - \! K m\big)+\left(\dfrac{K-t-1}{t+1}\right)\big( K m-t\big),   
\end{align}
where $t \in [K \!- \!1]$ and $t \leq K m \leq t \! + \!1 $.
\end{theorem}
\begin{proof}
\textbf{Achievability}: The centralized caching scheme proposed in \cite{ji2016fundamental} achieves (\ref{eqn_thm_equalcache}), which is also the optimal solution of (\ref{eqn_opt}). \textbf{Converse}: The proof is detailed in Section \ref{proof_thm_equalcache}.
\end{proof}

Next theorem shows that the heterogeneity in users cache sizes does not increase the achievable D2D delivery load as long as the smallest cache $m_1$ is large enough.

\begin{theorem}\label{thm_m1_lb}
For $N \geq K$, $m_1 \leq \dots \leq m_K$, and $m_1 \geq \dfrac{\sum_{k=2}^{K} m_k \! - \! 1 }{K \! - \! 2}$, the minimum worst-case D2D delivery load with uncoded placement,
\begin{align}\label{eqn_thm_m1_lb}
R^*_{\mathfrak{A}}(\bm m)=\left(\dfrac{K-t}{t}\right)\big( t \! + \! 1 \! - \! \sum_{k=1}^{K} m_k \big)+\left(\dfrac{K-t-1}{t+1}\right)\big( \sum_{k=1}^{K} m_k -t\big),   
\end{align}
where $t \leq \sum_{k=1}^{K} m_k \leq t \! + \!1$, and $\ t \in [K \!- \!1]$. 
\end{theorem}
\begin{proof}
\textbf{Achievability}: In Section \ref{proof_ach_thm_m1_lb}, we generalize the caching scheme in \cite{ji2016fundamental} to accommodate the heterogeneity in cache sizes. \textbf{Converse}: The proof is detailed in Section \ref{proof_conv_thm_m1_lb}.
\end{proof}
The next theorem characterizes the trade-off in the small memory regime defined as the total network cache memory is less than twice the library size. 

\begin{table}[t]
\small
\centering
\caption{Summary of the analytical results on $R^*_{\mathfrak{A}}(\bm m) $.  }\label{tab_main_results}
\begin{tabular}{|c|c|c|c|c|c|}
  \hline
  Regions & $ \sum\limits_{j=1}^{K} \! m_j \! \in \! [1,2]$ & $\dots$ & $\sum\limits_{j=1}^{K} \! m_j \! \in \! [t,t \! + \!1]  $ & $\dots $ & $\sum\limits_{j=1}^{K} \! m_j \! \in \! [K \! - \!1,K]  $\\
  \hline  
$(K\!-\!2)m_1 \! \geq \!  \sum\limits_{j=2}^{K} \!  m_j \!-\!1 $ & Exact (\ref{eqn_thm_m1_lb}) & $\dots$ & Exact (\ref{eqn_thm_m1_lb}) & $\dots$ & Exact (\ref{eqn_thm_m1_lb})\\
  \hline  
$(K\!-\!2)m_1 \! < \!  \sum\limits_{j=2}^{K} \!  m_j \!-\!1  $ &  Exact (\ref{eqn_thm_mtot_2}) &  & Achievability (\ref{eqn_opt}) &  & Exact (\ref{eqn_thm_mtot_K}) \\ 
$(K\!-\!3)m_2 \! \geq \!  \sum\limits_{j=3}^{K} \!  m_j \!-\!1  $ &  &  & Lower bound (\ref{eqn_bound_genie})& &\\
  \cline{1-1}
$\vdots$&  & $\dots$ &  & $\dots$ & \\
  \cline{1-1}  
$ (K\!-l\!-\!1) m_{l} \! < \! \! \! \sum\limits_{j=l+1}^{K} \! \!  m_j \!-\!1 $&  &  &  &  & \\
$\! (K\!-l\!-\!2) m_{l+1} \! \geq \! \! \! \sum\limits_{j=l+2}^{K} \! \! \! m_j \!-\!1 $ &  &  &  & &\\ \cline{1-1}  
$\vdots$&  &  &  &  & \\
 \cline{1-1}
$m_{K-2}\! < \! m_{K-1} \! + \! m_{K} \!-\!1 $ & &  &  &  & \\
 \hline 
\end{tabular}
\normalsize
\end{table} 
\begin{theorem}\label{thm_mtot_2}
For $N \geq K$, $m_1 \leq \dots \leq m_K$, $ 1 \leq \! \sum_{k=1}^{K} m_k \! \leq 2$, $m_l < \dfrac{\sum_{i=l+1}^{K} m_i \! - \! 1 }{K \! - \! l- \! 1}$ and $m_{l+1} \geq \dfrac{\sum_{i=l+2}^{K} m_i \! - \! 1 }{K \! - \! l- \! 2}$, the minimum worst-case D2D delivery load with uncoded placement,
\begin{align}\label{eqn_thm_mtot_2}
R^*_{\mathfrak{A}}(\bm m)= \dfrac{3K-l-2}{2}- \sum_{i=1}^{l} (K \! - \! i) m_i - \left(\dfrac{K-l}{2} \right)\sum_{i=l+1}^{K} m_i,   
\end{align}
where $ l \in [K \!- \!2]$. 
\end{theorem}
\begin{proof}
\textbf{Achievability}: The caching scheme is provided in Section \ref{proof_ach_thm_mtot_2}. \textbf{Converse}: The proof is detailed in Section \ref{proof_conv_thm_mtot_2}.
\end{proof}
From (\ref{eqn_thm_mtot_2}), we observe that the trade-off in the $l$th heterogeneity level depends on the individual cache sizes of users $\{1,\dots,l\}$ and the total cache sizes of the remaining users. 
\begin{remark}
The trade-off in the region where $\sum_{k=1}^{K} m_k \! \leq 2$ and $(K \! - \! 2) m_1 \geq \sum_{i=2}^{K} m_i  -  1$, which is included in Theorem \ref{thm_m1_lb}, can also be obtained by substituting $l\!=\!0$ in Theorem \ref{thm_mtot_2}.
\end{remark}

The next theorem characterizes the trade-off in the large memory regime defined as one where the total network memory satisfies $\sum_{k=1}^{K} m_k \! \geq K\!-\!1$. In particular, we show the optimality of uncoded placement, i.e., $R^*_{\mathfrak{A}}(\bm m)=R^*(\bm m) $. 

\begin{theorem}\label{thm_mtot_K}
For $N \geq K$, $m_1 \leq \dots \leq m_K$, and $ \sum_{k=1}^{K} m_k \! \geq K \! - \! 1$, the minimum worst-case D2D delivery load with uncoded placement,
\begin{align}\label{eqn_thm_mtot_K}
R^*_{\mathfrak{A}}(\bm m)= R^*(\bm m)=1  - m_1,   
\end{align}
where $m_1 < \dfrac{\sum_{i=2}^{K} m_i \! - \! 1 }{K \! - \! 2}$.
\end{theorem}
\begin{proof}
\textbf{Achievability}: The caching scheme is provided in Section \ref{proof_ach_thm_mtot_K}. \textbf{Converse}: The proof follows from the cut-set bound in \cite{ji2016fundamental}.
\end{proof}
Finally, for $K=3$, we have the complete characterization below.
\begin{theorem}\label{thm_3ue}
For $K=3$, $N \geq 3$, and $m_1 \leq m_2 \leq m_3$, the minimum worst-case D2D delivery load with uncoded placement,
\begin{align}\label{eqn_thm_3ue}
R^*_{\mathfrak{A}}(\bm m)= \max  \left\lbrace \! \frac{7}{2} \! - \!  \frac{3}{2} \big( m_1 \! + \! m_2 \! + \! m_3\big),   3 \! - \! 2 m_1 \! - \! m_2 \! - \! m_3,   \frac{3}{2} \! - \!  \frac{1}{2} \big( m_1 \! + \! m_2 \! + \! m_3\big),   1 \! - \! m_1 \! \right\rbrace \! \! .
\end{align}
\end{theorem}
\begin{proof}
\textbf{Achievability}: The proof is in Appendix \ref{app_thm_3ue_ach}. \textbf{Converse}: The proof is in Appendix \ref{app_thm_3ue_conv}.
\end{proof}

\vspace{-.08 in}
\section{General Caching Scheme}\label{sec_cach}
In the placement phase, we consider all feasible uncoded placement schemes in which the whole library can be retrieved utilizing the users' cache memories via D2D delivery, i.e., there must be no subfile stored at the server that is not placed in the end nodes in pieces. The delivery phase consists of $K$ transmission stages, in each of which one of the $K$ users acts as a ``server''. In particular, in the $j$th transmission stage, user $j$ transmits the signals $X_{j \rightarrow \mc T}$ to the users in the sets $\mc T \subsetneq_{\phi} [K] \setminus \{j\}$. \footnote{For convenience, we omit the subscript $ \bm d$ from $X_{j \rightarrow \mc T, \bm d}$ whenever the context is clear.} 
%

\subsection{Placement phase}
The server partitions each file $W_{n}$ into $2^K \! - \! 1$ subfiles, $\tilde W_{n,\mc S}, \mc S \subsetneq_{\phi} [K]$, such that $\tilde W_{n,\mc S}$ denotes a subset of $W_{n}$ which is stored exclusively at the users in the set $\mc S $. The partitioning is symmetric over the files, i.e., $|\tilde W_{n,\mc S}|= a_{\mc S} F \text{ bits}, \forall n \in [N]$, where the \textit{allocation variable} $a_{\mc S} \in [0,1]$ defines the size of $\tilde W_{n,\mc S}$ as a fraction of the file size $F$. Therefore, the set of feasible uncoded placement schemes, $\mathfrak A(\bm m)$, is defined by

\begin{equation}\label{eqn_feas_alloc}
 \mathfrak A(\bm m) = \Bigg\lbrace  \bm a \in [0,1]^{2^K} \Big\vert \sum\limits_{\mc S \subsetneq_{\phi}[K] } \! \! \! \! a_{\mc S }=1,  \! \! \! \! \! \! \! \! \sum\limits_{\mc S \subset [K]  \ \! : \ \! k \in \mc S } \! \! \! \! \! \! \! \! a_{\mc S } \leq m_k, \forall k \in [K] \Bigg\rbrace,\!\!
\end{equation}
where the allocation vector $\bm a $ consists of the allocation variables $a_{\mc S}, \mc S \subsetneq_{\phi}[K] $, the first constraint follows from the fact the whole library can be reconstructed from the users' cache memories, and the second represents the cache size constraint at user $k$. More specifically, user $k$ cache content is defined as
\vspace{-0.1in}
\begin{align}
Z_k = \bigcup\limits_{n \in [N]} \ \bigcup\limits_{\mc S \subset [K] \ \! : \ \! k \in \mc S } \tilde W_{n,\mc S}.
\end{align}

Next, we explain the delivery scheme for a three-user system for clarity of exposition, then we generalize to $K>3$.
\vspace{-0.07in}
\subsection{Delivery phase: Three-user system}
\subsubsection{Structure of $X_{j \rightarrow \mc T}$} 
In the first transmission stage, i.e., $j=1$, user $1$ transmits the unicast signals $X_{1 \rightarrow \{2\}}, X_{1 \rightarrow \{3\}}$, and the multicast signal $X_{1 \rightarrow \{2,3\}}$ to users $\{2,3\}$. In particular, the unicast signal $X_{1 \rightarrow \{2\}}$ delivers the subset of $W_{d_2} $ which is stored exclusively at user $1$, i.e., subfile $\tilde W_{d_2,\{1\}} $, in addition to a fraction of the subfile stored exclusively at users $\{1,3\} $, which we denote by $W_{d_2,\{1,3\}}^{1 \rightarrow \{2\}} $. In turn, $X_{1 \rightarrow \{2\}}$ is given by 
\begin{align}\label{eqn_delv_3UE_1}
X_{1 \rightarrow \{2\}}=\tilde W_{d_2,\{1\}} \bigcup W_{d_2,\{1,3\}}^{1 \rightarrow \{2\}},
\end{align} 
where $W_{d_2,\{1,3\}}^{1 \rightarrow \{2\}} \! \subset \! \tilde W_{d_2,\{1,3\}}$, such that $|W_{d_2,\{1,3\}}^{1 \rightarrow \{2\}}| \! = \! u_{\{1,3\}}^{1 \rightarrow \{2\}} F$ bits. That is, the \textit{assignment variable} $u_{\mc S}^{j \rightarrow \mc T} \in [0,a_{\mc S}]$ represents the fraction of the subfile $\tilde W_{\mc S} $ which is involved in the transmission from user $j$ to the users in $\mc T$. Similarly, the unicast signal $X_{1 \rightarrow \{3\}}$ is given by
\begin{align}\label{eqn_delv_3UE_2}
X_{1 \rightarrow \{3\}}=\tilde W_{d_3,\{1\}} \bigcup W_{d_3,\{1,2\}}^{1 \rightarrow \{3\}},
\end{align} 
where $W_{d_3,\{1,2\}}^{1 \rightarrow \{3\}} \! \subset \! \tilde W_{d_3,\{1,2\}}$, such that $|W_{d_3,\{1,2\}}^{1 \rightarrow \{3\}}| \! = \! u_{\{1,2\}}^{1 \rightarrow \{3\}} F$ bits.

The multicast signal $X_{1 \rightarrow \{2,3\}}$ is created by XORing the pieces $W_{d_2,\{1,3\}}^{1 \rightarrow \{2,3\}}$, and $W_{d_3,\{1,2\}}^{1 \rightarrow \{2,3\}}$, which are assumed to have equal size. That is, $X_{1 \rightarrow \{2,3\}}$ is defined by
\begin{align}\label{eqn_delv_3UE_3}
X_{1 \rightarrow \{2,3\}}=W_{d_2,\{1,3\}}^{1 \rightarrow \{2,3\}} \oplus  W_{d_3,\{1,2\}}^{1 \rightarrow \{2,3\}},
\end{align}
where $W_{d_2,\{1,3\}}^{1 \rightarrow \{2,3\}} \! \subset \! \tilde W_{d_2,\{1,3\}}$ and $W_{d_3,\{1,2\}}^{1 \rightarrow \{2,3\}} \! \subset \! \tilde W_{d_3,\{1,2\}}$. 

From (\ref{eqn_delv_3UE_1})-(\ref{eqn_delv_3UE_3}), we observe that subfile $\tilde W_{d_2,\{1,3\}}$ contributes to both $X_{1 \rightarrow \{2\}}$, and $X_{1 \rightarrow \{2,3\}} $. Additionally, in the third transmission stage subfile $\tilde W_{d_2,\{1,3\}}$ contributes to both $X_{3 \rightarrow \{2\}}$, and $X_{3 \rightarrow \{1,2\}} $. Therefore, in order to prevent users $\{1,3\}$ from transmitting redundant bits to user $2$ from $\tilde W_{d_2,\{1,3\}}$, we need to partition $\tilde W_{d_2,\{1,3\}}$ into disjoint pieces, i.e., 
\begin{align}\label{eqn_delv_3UE_4}
 W_{d_2,\{1,3\}}^{1 \rightarrow \{2\}} \bigcap W_{d_2,\{1,3\}}^{1 \rightarrow \{2,3\}} \bigcap W_{d_2,\{1,3\}}^{3 \rightarrow \{2\}} \bigcap W_{d_2,\{1,3\}}^{3 \rightarrow \{1,2\}} = \phi. 
\end{align}

\subsubsection{Delivery phase constraints} Next, we describe the delivery phase in terms of linear constraints on the transmission variables $v_{j \rightarrow \mc T} $ and the assignment variables $u_{\mc S}^{j \rightarrow \mc T} $, which represent $|X_{j \rightarrow \mc T}|/F  $ and  $|W_{d_i,\mc S}^{j \rightarrow \mc T}|/F $, respectively. 
 First, the structure of the unicast signals in (\ref{eqn_delv_3UE_1}) and (\ref{eqn_delv_3UE_2}) is represented by 
\begin{align}
v_{1 \rightarrow \{2\}}=a_{\{1\}}+ u_{\{1,3\}}^{1 \rightarrow \{2\}}, \label{eqn_delv_3UE_5}\ 
v_{1 \rightarrow \{3\}}=a_{\{1\}}+ u_{\{1,2\}}^{1 \rightarrow \{3\}}. 
\end{align}
Similarly, for the second and third transmission stage, we have
\begin{align}
v_{2 \rightarrow \{1\}}=a_{\{2\}}+ u_{\{2,3\}}^{2 \rightarrow \{1\}}, \
v_{2 \rightarrow \{3\}}=a_{\{2\}}+ u_{\{1,2\}}^{2 \rightarrow \{3\}}, \\ 
v_{3 \rightarrow \{1\}}=a_{\{3\}}+ u_{\{2,3\}}^{3 \rightarrow \{1\}}, \
v_{3 \rightarrow \{2\}}=a_{\{3\}}+ u_{\{1,3\}}^{3 \rightarrow \{2\}}. 
\end{align}
The structure of the multicast signal in (\ref{eqn_delv_3UE_3}) is represented by
\begin{align}
v_{1 \rightarrow \{2,3\}}=u_{\{1,3\}}^{1 \rightarrow \{2,3\}} = u_{\{1,2\}}^{1 \rightarrow \{2,3\}}.  
\end{align}
Similarly, for the second and third transmission stage, we have
\begin{align}
v_{2 \rightarrow \{1,3\}}=u_{\{2,3\}}^{1 \rightarrow \{2,3\}} = u_{\{1,2\}}^{1 \rightarrow \{2,3\}}, \
v_{3 \rightarrow \{1,2\}}=u_{\{2,3\}}^{3 \rightarrow \{1,2\}} = u_{\{1,3\}}^{3 \rightarrow \{1,2\}}.  
\end{align}
Additionally, (\ref{eqn_delv_3UE_4}) ensures that $\tilde W_{d_2,\{1,3\}}$ is divided into disjoint pieces which prevents the transmission of redundant bits. Hence, we have 
\begin{align}\label{eqn_delv_3UE_7}
u_{\{1,3\}}^{1 \rightarrow \{2\}} + u_{\{1,3\}}^{1 \rightarrow \{2,3\}} + u_{\{1,3\}}^{3 \rightarrow \{2\}} + u_{\{1,3\}}^{3 \rightarrow \{1,2\}} \leq a_{\{1,3\}}.  
\end{align}
Similarly, the redundancy constraints for the subfiles $\tilde W_{d_3,\{1,2\}}$ and $\tilde W_{d_1,\{2,3\}}$, are given by
\begin{align}
u_{\{1,2\}}^{1 \rightarrow \{3\}} + u_{\{1,2\}}^{1 \rightarrow \{2,3\}} + u_{\{1,2\}}^{2 \rightarrow \{3\}} + u_{\{1,2\}}^{2 \rightarrow \{1,3\}} \leq a_{\{1,2\}}, \\ 
u_{\{2,3\}}^{2 \rightarrow \{1\}} + u_{\{2,3\}}^{2 \rightarrow \{1,3\}} + u_{\{2,3\}}^{3 \rightarrow \{1\}} + u_{\{2,3\}}^{3 \rightarrow \{1,2\}} \leq a_{\{2,3\}}.  \label{eqn_delv_3UE_8}
\end{align}
Finally, we need to ensure that the transmitted signals complete the requested files, i.e., we have the following delivery completion constraints
\begin{align}
v_{2 \rightarrow \{1\}}+v_{2 \rightarrow \{1,3\}}+v_{3 \rightarrow \{1\}}+v_{3 \rightarrow \{1,2\}} \geq 1-\sum_{\mc S \subset [K]  \ \! : \ \! 1 \in \mc S } \! \! \! \! \! \! \! \! a_{\mc S }, \label{eqn_delv_3UE_9} \\
v_{1 \rightarrow \{2\}}+v_{1 \rightarrow \{2,3\}}+v_{3 \rightarrow \{2\}}+v_{3 \rightarrow \{1,2\}} \geq 1-\sum\limits_{\mc S \subset [K]  \ \! : \ \! 2 \in \mc S } \! \! \! \! \! \! \! \! a_{\mc S }, \\
v_{1 \rightarrow \{3\}}+v_{1 \rightarrow \{2,3\}}+v_{2 \rightarrow \{3\}}+v_{2 \rightarrow \{1,3\}} \geq 1-\sum\limits_{\mc S \subset [K]  \ \! : \ \! 3 \in \mc S } \! \! \! \! \! \! \! \! a_{\mc S }.\label{eqn_delv_3UE_6}
\end{align}
Therefore, the set of feasible linear delivery schemes for a three-user system is defined by (\ref{eqn_delv_3UE_5})-(\ref{eqn_delv_3UE_6}), and $u_{\mc S}^{j \rightarrow \mc T} \! \! \in [0, a_{\mc S} ]$.  
\subsection{Delivery phase: $K$-user system}
In general, the unicast signal transmitted by user $j$ to user $i $ is defined by
\begin{align}\label{eqn_delv_KUE_1}
X_{j \rightarrow \{i\}}=\tilde W_{d_i,\{j\}} \bigcup \Bigg( \bigcup_{\mc S \subset [K]\setminus \{i\} \! \ : \! \ j \in \mc S, |\mc S| \geq 2 } \! \! \! \! \! \! \! \! \! \! \! \! \! \! \! \! \! W_{d_i,\mc S}^{j \rightarrow \{i\}} \Bigg),
\end{align}
where $ W_{d_i,\mc S}^{j \rightarrow \{i\}} \! \subset \! \tilde W_{d_i,\mc S} $ such that $|W_{d_i,\mc S}^{j \rightarrow \{i\}}| = u^{ j \rightarrow \{i\}}_{\mc S} F$ bits. While, user $j$ constructs the multicast signal $X_{j \rightarrow \mc T}$, such that the piece intended for user $i \in \mc T $, which we denote by $W_{d_i}^{j \rightarrow \mc T} $, is stored at users $\{j\} \cup (\mc T \setminus \{i\}) $. That is, $X_{j \rightarrow \mc T}$ is constructed using the side information at the sets
\begin{align}
\mc B^{ j \rightarrow \mc T}_{i } \triangleq \Big\{ \mc S \subset [K] \setminus \{i\} :  \{j\} \cup (\mc T \setminus \{i\}) \subset \mc S \Big\}, 
\end{align}
which represents the subfiles stored at users $\{j\} \cup (\mc T \setminus \{i\}) $ and not available at user $i \in \mc T$. In turn, we have
\begin{align}\label{eqn_delv_KUE_2}
X_{j \rightarrow \mc T}= \oplus_{i \in \mc T} W_{d_i}^{j \rightarrow \mc T} = \oplus_{i \in \mc T} \bigg( \bigcup_{\mc S \in \mc B^{ j \rightarrow \mc T}_{i } } \! \! \! \!  W_{d_i,\mc S}^{j \rightarrow \mc T} \bigg).
\end{align}
\begin{remark}\label{remark_multicast} The definition of the multicast signals in (\ref{eqn_delv_KUE_2}) allows flexible utilization of the side-information, i.e., $X_{j \rightarrow \mc T} $ is not defined only in terms of the side-information stored exclusively at users $\{j\} \cup (\mc T \setminus \{i\}) $ as in \cite{ji2016fundamental}. Furthermore, a delivery scheme with the multicast signals $X_{j \rightarrow \mc T}= \oplus_{i \in \mc T} W_{d_i,\{j\} \cup (\mc T \setminus \{i\})}^{j \rightarrow \mc T} $ is suboptimal in general.
\end{remark}
%
%
\begin{algorithm}[t]
\begin{algorithmic}[1]
\REQUIRE $\bm d, \bm a, \bm u, \bm v$, and $\tilde W_{n,\mc S}$ 
\ENSURE $ X_{j \rightarrow \mc T }, \! \ \forall j \in [K], \! \ \forall \ \! \mc T \subsetneq_{\phi} [K] \! \setminus \! \{j\}$

\COMMENT{Partitioning}
\FOR{$\{ \mc S \subset [K] : 2 \leq |\mc S|\leq K \! - \! 1 \}$}
\FOR{$ \{i \in [K]: i \not\in \mc S\}$}
\STATE Divide $\tilde W_{d_i,\mc S}$ into $ W_{d_i,\mc S}^{j \rightarrow \mc T}, \! \ \forall j \in \mc S, \forall \! \  \mc T \subset \{i\} \cup (\mc S \setminus \{j\}) \text{ s.t. } i \in \mc T $, such that $| W_{d_i,\mc S}^{j \rightarrow \mc T}|=u^{j \rightarrow \mc T}_{\mc S} F$ bits.  
\ENDFOR
\ENDFOR

\COMMENT{Transmission stage $j$}
\FOR{$ j \in [K]$}  
\FOR{$ \mc T \subsetneq_{\phi} [K] \! \setminus \! \{j\}$}  
\IF{$\mc T=\{i\} $} 
\STATE $X_{j \rightarrow \{i\}} \leftarrow  \tilde W_{d_i,\{j\}} \bigcup \bigg( \bigcup\limits_{\mc S \subset [K]\setminus \{i\} \! \ : \! \ j \in \mc S, |\mc S| \geq 2 } \! \! \! \! \! \! \! \! \! \! \! \!   W_{d_i,\mc S}^{j \rightarrow \{i\}} \bigg)$ 
\ELSE 
\STATE $X_{j \rightarrow \mc T} \leftarrow \oplus_{i \in \mc T} \bigg( \bigcup\limits_{\mc S \in \mc B^{ j \rightarrow \mc T}_{i } } \! \! \! \!  W_{d_i,\mc S}^{j \rightarrow \mc T} \bigg) $ 
\ENDIF
\ENDFOR
\ENDFOR
\end{algorithmic}
 \caption{D2D delivery procedure}\label{alg_delv}
\end{algorithm}
%
The set of feasible linear delivery schemes, $\mathfrak{D}(\bm a) $, is defined by 
\begin{align}
& \! \! \!    v_{j \rightarrow \{i\}} \! = \! a_{\{j\}} + \! \! \!  \! \! \sum_{\mc S \subset [K]\setminus \{i\} \ \! : \ \!  j \in \mc S, |\mc S|\geq 2 }   \! \! \! \! \! \! \! \! \! \! \!   u^{j \rightarrow \{i\}}_{\mc S}, \! \! \ \forall j \in [K], \! \ \forall \! \  i \in \mc T, \label{eqn_dlv_cnst1} \\
& \! \! \!    v_{j \rightarrow \mc T} \! = \! \! \! \! \! \! \sum_{\mc S \in \mc B_i^{j \rightarrow \mc T } } \! \!  u^{j \rightarrow \mc T}_{\mc S}, \! \! \ \forall j \in [K], \! \ \forall \ \! \mc T \subsetneq_{\phi} [K] \! \setminus \! \{j\} \! , \! \ \forall   i \in \mc T, \!\label{eqn_dlv_cnst2} \\
&\sum_{j \in \mc S} \sum\limits_{\mc T \subset \{i\} \cup (\mc S \setminus \{j\}) \ \! : \ \! i \in \mc T} \! \! \! \! \! \! \! \! \! \! \!  \! \! \!  u^{j \rightarrow \mc T }_{\mc S} \leq a_{\mc S}, \! \ \forall \! \  i \not\in \mc S,  \forall \! \ \mc S \subset [K] \text{ s.t. }  2 \leq |\mc S| \leq K \! - \! 1, \label{eqn_dlv_cnst3}  \\
&\sum_{j \in [K] \setminus \{k\}} \sum\limits_{\mc T \subset [K] \setminus \{j\} \ \! : \ \! k \in \mc T } \! \! \! \! v_{ j \rightarrow \mc T }  \geq 1- \! \! \! \! \sum\limits_{\mc S \subset [K]  \ \! : \ \! k \in \mc S } \! \! \! \! \! \! \! \! a_{\mc S }, \ \forall \! \  k \in [K], \label{eqn_dlv_cnst4} \\
& 0 \leq u_{\mc S}^{ j \rightarrow \mc T } \! \! \leq a_{\mc S}, \forall j \! \in \! [K],  \forall \! \  \mc T \! \subsetneq_{\phi} \! [K] \! \setminus \! \{j\},  \forall \! \  \mc S \! \in \! \mc B^{j \rightarrow \mc T} \! \! \! , \label{eqn_dlv_cnst5}  
\end{align}
where $\mc B^{j \rightarrow \mc T} \triangleq \bigcup_{ i \in \mc T} \mc B^{ j \rightarrow \mc T}_{i }$. Note that (\ref{eqn_dlv_cnst1}) follows from the structure of the unicast signals in (\ref{eqn_delv_KUE_1}), (\ref{eqn_dlv_cnst2}) follows from the structure of the multicast signals in (\ref{eqn_delv_KUE_2}), (\ref{eqn_dlv_cnst3}) generalizes the redundancy constraints in (\ref{eqn_delv_3UE_7})-(\ref{eqn_delv_3UE_8}), and (\ref{eqn_dlv_cnst4}) generalizes the delivery completion constraints in (\ref{eqn_delv_3UE_9})-(\ref{eqn_delv_3UE_6}). The delivery procedure is summarized in Algorithm \ref{alg_delv}.

\subsection{Examples}

Next, we illustrate the solution of (\ref{eqn_opt}) by two examples. The first example considers a case where the delivery scheme in \cite{ji2016fundamental} can be generalized for unequal cache sizes, by considering D2D transmissions with different sizes. In turn, the heterogeneity in cache sizes does not increase the delivery load, i.e., we achieve the same delivery load in a homogeneous system with the same aggregate cache size.    
\begin{example}
For $K \!= \! N \! = \! 3 $ and $\bm m =[0.6, \ \! 0.7, \ \! 0.8]$, the optimal caching scheme is as follows: \newline
\textbf{Placement phase:} Each file $W_{n}$ is divided into four subfiles, such that $a_{\{1,2\}} \! = \! 0.2$, $a_{\{1,3\}} \! = \! 0.3$, $a_{\{2,3\}} \! = \! 0.4$, and $a_{\{1,2,3\}} \! = \! 0.1$.\newline 
\textbf{Delivery phase:} The D2D transmissions 
\begin{itemize}
\item $| X_{1 \rightarrow \{2,3\}}|/F\!=\!$ $v_{1 \rightarrow \{2,3\}} \! = \! u_{\{1,2\}}^{1 \rightarrow \{2,3\}} \! = \! u_{\{1,3\}}^{1 \rightarrow \{2,3\}} \! = \! 0.05 .$
\item $| X_{2 \rightarrow \{1,3\}}|/F\! =\!$ $v_{2 \rightarrow \{1,3\}} \! = \! u_{\{1,2\}}^{2 \rightarrow \{1,3\}} \! = \! u_{\{2,3\}}^{2 \rightarrow \{1,3\}} \! = \! 0.15. $
\item $|X_{3 \rightarrow \{1,2\}}|/F\! =\!$ $v_{3 \rightarrow \{1,2\}} \! = \! u_{\{1,3\}}^{3 \rightarrow \{1,2\}} \! = \! u_{\{2,3\}}^{3 \rightarrow \{1,2\}} \! = \! 0.25. $
\end{itemize}
The placement and delivery phases are illustrated in Fig. \ref{fig:D2D_ex}. Note that the same delivery load is achieved by the caching scheme in \cite{ji2016fundamental} for $\bm m =[0.7, \ \! 0.7, \ \! 0.7]$. In Theorem \ref{thm_3ue}, we show that the proposed scheme achieves $R^*_{\mathfrak{A}} (\bm m) = \dfrac{3}{2}  - \dfrac{1}{2} \big( m_1+m_2+m_3\big) \! = \! 0.45$.
\end{example}
\begin{figure}[t]
\includegraphics[scale=1.1]{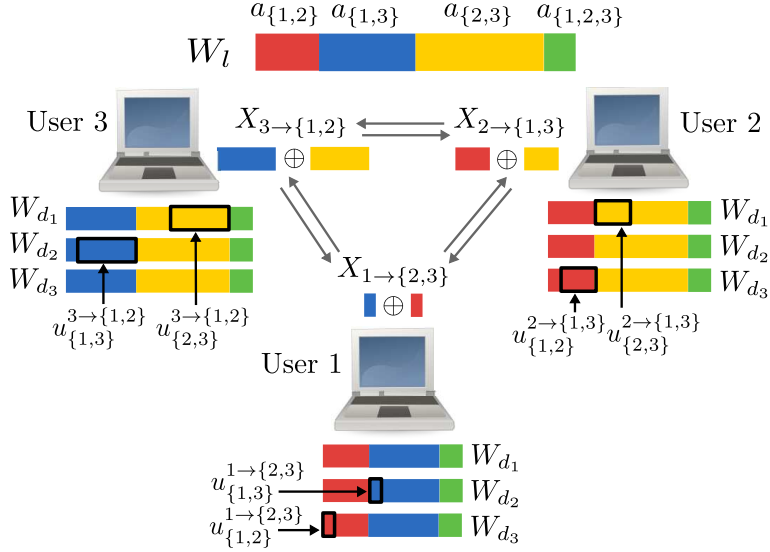}
\centering
\vspace{-.1in}
\caption{Example $K=N=3$, and $ \bm m =[0.6, \ \! 0.7, \ \! 0.8]$.}\label{fig:D2D_ex}
\end{figure}

Next example shows the suboptimality of delivery schemes that do not allow flexible utilization of the side-information, as pointed out in Remark \ref{remark_multicast}. By contrast, our delivery scheme achieves the delivery load memory trade-off with uncoded placement, $R^*_{\mathfrak{A}} (\bm m)$. 
\begin{example}
For $K \!= \! N \! = \! 4 $ and $\bm m =[0.2, \ \! 0.7, \ \! 0.7, \ \! 0.7]$, we have $R^*_{\mathfrak{A},\mathfrak{D}} (\bm m)\! =R^*_{\mathfrak{A}} (\bm m) \! = \!1.05$, and the optimal caching scheme is as follows: \newline
\textbf{Placement phase:} Each file $W_n$ is divided into seven subfiles, such that $a_{\{1,2\}} \! = \! a_{\{1,3\}} \! = \! a_{\{1,4\}} \! = \! 0.2/3$, $a_{\{2,3\}} \! = \! a_{\{2,4\}} \! = \! a_{\{3,4\}} \! = \! 0.5/3$, and $a_{\{2,3,4\}} \! = \! 0.3$.\newline 
\textbf{Delivery phase:} We have the D2D transmissions $X_{2 \rightarrow \{1\}}$, $X_{2 \rightarrow \{1,3\}}$, $X_{2 \rightarrow \{1,4\}}$, $X_{2 \rightarrow \{3,4\}}$, $X_{3 \rightarrow \{1\}}$, $X_{3 \rightarrow \{1,2\}}$, $X_{3 \rightarrow \{1,4\}}$, $X_{3 \rightarrow \{2,4\}}$, $X_{4 \rightarrow \{1\}}$, $X_{4 \rightarrow \{1,2\}}$, $X_{4 \rightarrow \{1,3\}}$, and $X_{4 \rightarrow \{2,3\}}$. In particular, we have $v_{2 \rightarrow \{1\}} \! = \! v_{3 \rightarrow \{1\}} \! = \! v_{4 \rightarrow \{1\}} \! = \! 0.4/3$, $v_{2 \rightarrow \{1,3\}} \! = \! v_{2 \rightarrow \{1,4\}} \! = \! v_{3 \rightarrow \{1,2\}} \! = \! v_{3 \rightarrow \{1,4\}} \! = \! v_{4 \rightarrow \{1,2\}} \! = \! v_{4 \rightarrow \{1,3\}} \! = \! 0.2/3$, and $v_{2 \rightarrow \{3,4\}} \! = \! v_{3 \rightarrow \{2,4\}} \! = \! v_{4 \rightarrow \{2,3\}} \! = \! 0.25/3$. More specifically, the signals transmitted by user $2$ are defined as follows
\begin{itemize}
\item $| X_{2 \rightarrow \{1\}}|/F\!=\!$ $v_{2 \rightarrow \{1\}} \! = \! u_{\{2,3\}}^{2 \rightarrow \{1\}} \! + \! u_{\{2,4\}}^{2 \rightarrow \{1\}} \! + \! u_{\{2,3,4\}}^{2 \rightarrow \{1\}}  \!= \! (0.05+0.05+0.3)/3.$
\item $| X_{2 \rightarrow \{1,3\}}|/F\! =\!$ $v_{2 \rightarrow \{1,3\}} \! = \! u_{\{1,2\}}^{2 \rightarrow \{1,3\}} \! = \! u_{\{2,3\}}^{2 \rightarrow \{1,3\}} \! = \! 0.2/3. $
\item $|X_{2 \rightarrow \{1,4\}}|/F\! =\!$ $v_{2 \rightarrow \{1,4\}} \! = \! u_{\{1,2\}}^{2 \rightarrow \{1,4\}} \! = \! u_{\{2,4\}}^{2 \rightarrow \{1,4\}} \! = \! 0.2/3. $
\item $|X_{2 \rightarrow \{3,4\}}|/F\! =\!$ $v_{2 \rightarrow \{3,4\}} \! = \! u_{\{2,3\}}^{2 \rightarrow \{3,4\}} \! = \! u_{\{2,4\}}^{2 \rightarrow \{3,4\}} \! = \! 0.25/3. $
\end{itemize}
Note that the signals transmitted by users $3$ and $4$ have similar structure to the signals transmitted by user $2$, which are illustrated in Fig. \ref{fig:D2D_ex2}. If we restrict the design of the D2D signals to be in the form of  $X_{j \rightarrow \mc T}= \oplus_{i \in \mc T} W_{d_i,\{j\} \cup (\mc T \setminus \{i\})}^{j \rightarrow \mc T} $, i.e., without the flexibility in utilizing the side information, we achieve a delivery load equal to $1.6$ compared with the optimal load $R^*_{\mathfrak{A}} (\bm m)\! = \!1.05$. 
\end{example}

\begin{figure}[t]
\includegraphics[scale=1]{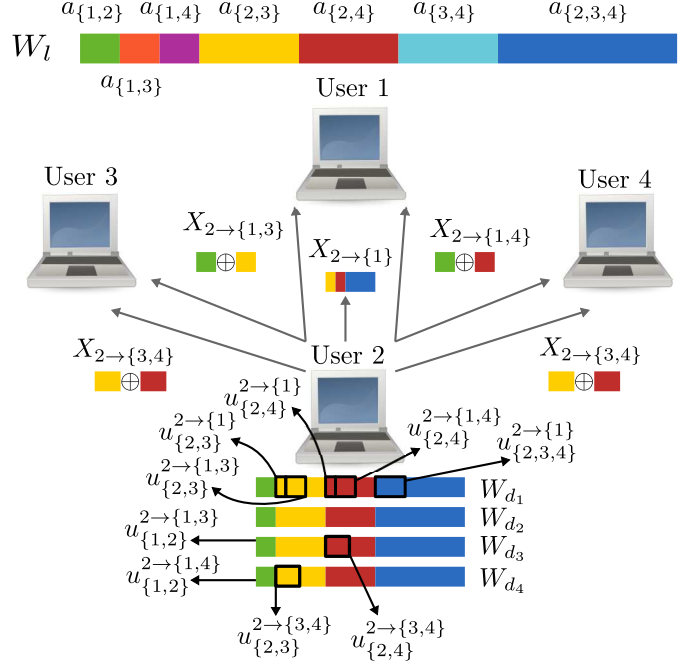}
\centering
\vspace{-0.1in}
\caption{Example $K=N=4$, and $ \bm m =[0.2, \ \! 0.7, \ \! 0.7, \ \! 0.7]$.}\label{fig:D2D_ex2}
\end{figure} 



\section{ Caching scheme: Achievability }\label{sec_achiev}

Next, we explicitly define the caching schemes that achieve the delivery loads defined in Theorems \ref{thm_m1_lb}, \ref{thm_mtot_2}, and \ref{thm_mtot_K}. 

\subsection{Achievability Proof of Theorem \ref{thm_m1_lb} }\label{proof_ach_thm_m1_lb}%
Next, we explain how the caching scheme in \cite{ji2016fundamental} can be tailored to systems with unequal cache sizes. Recall that for a homogeneous system where $m_k \! = \! m, \ \! \forall k$, in the placement phase, $W_{n}$ is divided into subfiles 
$\tilde W_{n,\mc S}, \mc S \subset [K]$, where $|\mc S| \in \{t,t+1\}$ for $t \leq \sum_{k=1}^{K} m_k \leq t  + 1$ and $\ t \in [K \!- \!1]$ \cite{ji2016fundamental}. More specifically, subfiles stored at the same number of users have equal size, i.e., $|\tilde W_{n,\mc S}|= |\tilde W_{n,\mc S^\prime}|$ if $|S|=|S^\prime|$. In order to accommodate the heterogeneity in cache sizes, we generalize the placement scheme in \cite{ji2016fundamental}, by allowing subfiles stored at the same number of users to have different sizes. The delivery procedure in \cite{ji2016fundamental} is generalized as follows. First, we further divide $\tilde W_{d_i,\mc S}$ into $|\mc S|$ pieces, $ W_{d_i,\mc S}^{j \rightarrow \mc S \setminus \{j\}\cup \{i\}}, \ \! j \in \mc S$, such that
\begin{align}
\left|W_{d_i,\mc S}^{j \rightarrow \mc S \setminus \{j\}\cup \{i\}} \right| = \begin{cases} \eta_{j} F, \text{ if } |\mc S|=t.\\
\theta_{j} F, \text{ if } |\mc S|=t+1.
\end{cases}
\end{align}
The multicast signal $X_{j \rightarrow \mc T}$ is constructed such that the piece requested by user $i$ is cached by the remaining $\mc T \setminus \{i\} $ users. That is, user $j$ transmits the signals $X_{j \rightarrow \mc T} = \oplus_{i \in \mc T} W^{j \rightarrow \mc T}_{d_i,\{j\} \cup \mc T \setminus \{i\}}$, $\forall \mc T \subset [K]\setminus \{j\}$ and $|\mc T| \in \{t,t+1\}$. For example, for $K=4$ and $t=2$, we have  
\begin{align}
X_{j \rightarrow \{i_1,i_2\}} &=  W^{j \rightarrow \{i_1,i_2\}}_{d_{i_1},\{j,i_2\}} \oplus  W^{j \rightarrow \{i_1,i_2\}}_{d_{i_2},\{j,i_1\}}, \\
X_{j \rightarrow \{i_1,i_2,i_3\}} &=  W^{j \rightarrow \{i_1,i_2,i_3\}}_{d_{i_1},\{j,i_2,i_3\}} \oplus W^{j \rightarrow \{i_1,i_2,i_3\}}_{d_{i_2},\{j,i_1,i_3\}} \oplus W^{j \rightarrow \{i_1,i_2,i_3\}}_{d_{i_3},\{j,i_1,i_2\}}.
\end{align}
In turn, the D2D delivery load is given as 
\begin{align}\label{eqn_app_24}
R^*_{\mathfrak{A},\mathfrak{D}}(\bm m)= \binom{K-1}{t} \sum_{j=1}^{K} \eta_{j} + \binom{K-1}{t+1} \sum_{j=1}^{K} \theta_{j}.
\end{align}

Next, we need to choose $\eta_{j}$ and $\theta_{j}$ taking into account the feasibility of the placement phase. To do so, we need to choose a non-negative solution to the following equations
\begin{align}
\binom{K-1}{t-1} \eta_k + \binom{K-2}{t-2} \! \! \sum_{i \in [K]\setminus \{k\}} \! \! \! \! \eta_i + \binom{K-1}{t} \theta_k + \binom{K-2}{t-1} \! \! \sum_{i \in [K]\setminus \{k\}} \! \! \! \! \theta_i &= m_k, \ \! \forall k \in [K], \\
\binom{K-1}{t-1} \! \! \sum_{i \in [K]} \eta_i + \binom{K-1}{t} \! \! \sum_{i \in [K]}  \theta_i &= 1,
\end{align}
which can be simplified to 
\begin{align}
\sum_{i=1}^{K} \eta_i &= \dfrac{t+1- \sum_{i=1}^{K} m_k}{\binom{K-1}{t-1}}, \label{eqn_app_25} \\
\eta_k + \dfrac{K \! - \! t \! - \!1}{t} \theta_k &= \dfrac{1+(K \! - \!2)m_k - \sum_{i \in [K]\setminus \{k\}} m_i}{(K \! - \! t) \binom{K-1}{t-1}}, \ \! \forall k \in [K], \label{eqn_app_26}
\end{align}
By combining (\ref{eqn_app_24}), (\ref{eqn_app_25}), and (\ref{eqn_app_26}), one can show that the D2D delivery load is given as 
\begin{align}
R^*_{\mathfrak{A},\mathfrak{D}}(\bm m)=\left(\dfrac{K-t}{t}\right)\left( t \! + \! 1 \! - \! \sum_{k=1}^{K} m_k \right)+\left(\dfrac{K-t-1}{t+1}\right)\left( \sum_{k=1}^{K} m_k -t\right).   
\end{align}
Observe that there always exists a non-negative solution to (\ref{eqn_app_25}) and (\ref{eqn_app_26}), since we have $(K \! - \! 2) m_1 \geq \sum_{k=2}^{K} m_k \! - \! 1 $. For instance, one can assume that $\binom{K-1}{t-1} \eta_k = \rho_k \big( t+1- \sum_{i=1}^{K} m_i \big)$, where $ \sum_{k=1}^{K} \rho_k =1$ and $0 \leq \rho_k \leq \dfrac{1+(K \! - \!2)m_k - \sum_{i \in [K]\setminus \{k\}} m_i}{(K \! - \! t)(t \! + \! 1 \! - \! \sum_{i=1}^{K} m_k)}$, which guarantee that $\eta_k,\theta_k \geq 0$.
\begin{remark} For nodes with equal cache sizes, the proposed scheme reduces to the scheme proposed in \cite{ji2016fundamental}. In particular, for $m_k=t/K, \forall k $, we get $\theta_j=0, \forall j $ and $\eta_j = \dfrac{1}{t \binom{K}{t}}, \forall j$.
\end{remark}

\subsection{Achievability Proof of Theorem \ref{thm_mtot_2} }\label{proof_ach_thm_mtot_2}%
For $(K \! - \! l- \! 1)m_l < \sum_{i=l+1}^{K} m_i \! - \! 1$ and $(K \! - \! l- \! 2)m_{l+1} \geq \sum_{i=l+2}^{K} m_i \! - \! 1$, where $ l \in [K \!- \!2]$, in the placement phase, each file $W_{n}$ is partitioned into subfiles $\tilde W_{n,\{i\}}, i \in \{l\!+\!1,\dots,K\} $, $\tilde W_{n,\{j,i\}}, j \in [l],i \in \{l\!+\!1,\dots,K\}  $, and $ \tilde W_{n,\mc S}, \mc S \subset  \{l\!+\!1,\dots,K\}, |\mc S|\!= \!2$, which satisfy 
\begin{subequations}\label{eqn_achv_thm5_1}
\begin{align}
&\sum_{j=l+1}^{K} a_{\{j\}} = 2- \sum_{k=1}^{K} m_k, \ \sum_{\mc S \subset \{l+1,\dots,K\} : |\mc S|=2}  \! \! \! \! \! \! \! \! a_{\mc S} = \sum_{i=l+1}^{K} m_i -1,\\
&\sum_{j=l+1}^{K} a_{\{i,j\}} = m_j, \ \! i \in [l], \ a_{\{j\}} + \sum_{\mc S \subset [K] : |\mc S|=2, j \in \mc S} \! \! \! \! \! \! \! \! a_{\mc S} =m_j, \ \! j=l\!+\!1,\dots,K.
\end{align}
\end{subequations}
In particular, we choose any non-negative solution to (\ref{eqn_achv_thm5_1}) that satisfies
\begin{enumerate}
\item For $j \in \{l+1,\dots,K\} $, $a_{\{i_1,j\}} \leq a_{\{i_2,j\}}$ if $i_1 < i_2$, which is feasible because $m_{i_1}\leq m_{i_2} $.
\item For $\{i,j\}\subset \{l\!+\!1, \dots, K\}$, $a_{\{l,i\}} + a_{\{l,j\}} \leq a_{\{i,j\}}$, which is also feasible because $(K \! - \! l- \! 1)m_l < \sum_{i=l+1}^{K} m_i \! - \! 1$.
\end{enumerate}
In the delivery phase, we have the following multicast transmissions:
\begin{itemize}
\item \underline{Multicast to user $1$}: For $j \in \{l\!+\!1,\dots,K \}$ and $i \in [K]\setminus \{1,j\} $, we choose $v_{j \rightarrow \{1,i\}} \! = \! a_{\{1,j\}}$. 
\begin{align}
\sum_{j=l+1}^{K} \sum_{i \in [K]\setminus \{1,j\}} v_{j \rightarrow \{1,i\}} =\sum_{j=l+1}^{K} \sum_{i \in [K]\setminus \{1,j\}} a_{\{1,j\}} = (K-2)m_1.
\end{align}
\item \underline{Multicast to user $2$}: For $j \in \{l\!+\!1,\dots,K \}$ and $i \in [K]\setminus \{1,2,j\} $, we choose $v_{j \rightarrow \{2,i\}} \! = \! a_{\{2,j\}}$.
\begin{align}
\! \! \! \! \! \! \sum_{j=l+1}^{K} \! \! v_{j \rightarrow \{1,2\}} + \sum_{j=l+1}^{K} \sum_{i \in [K]\setminus \{1,2,j\}} \! \! \! \!\! \! \! \! v_{j \rightarrow \{2,i\}} \! = \! \sum_{j=l+1}^{K} \! \! a_{\{1,j\}} \! + \! \sum_{j=l+1}^{K} \sum_{i \in [K]\setminus \{1,2,j\}} \! \! \! \! \! \! \! \! a_{\{2,j\}} \! = \! m_1 \! + \! (K \! - \! 3)m_2.
\end{align}
\item \underline{Multicast to user $k \in \{3, \dots, l\}$}: Similarly, we have
\begin{align}
\! \! \! \! &\sum_{j=l+1}^{K} v_{j \rightarrow \{1,l\}} + \dots + \sum_{j=l+1}^{K} v_{j \rightarrow \{k-1,l\}} +\sum_{j=l+1}^{K} \sum_{i \in [K]\setminus \{1,\dots,k,j\}} v_{j \rightarrow \{l,i\}} \\
&= \! \! \sum_{j=l+1}^{K} \! a_{\{1,j\}} \! + \dots + \! \sum_{j=l+1}^{K} \! a_{\{k-1,j\}} \! + \! \sum_{j=l+1}^{K} \sum_{i \in [K]\setminus \{1,\dots,k,j\}} \! \! \! \! \! \! \! \! a_{\{l,j\}} = \sum_{i=1}^{k-1} m_i+(K\! - \! k \! - \! 1)m_l.
\end{align}
\item \underline{Multicast to users $\{l\!+\!1, \dots, K\}$}: For $\{i_1,i_2\} \subset \{l\!+\!1, \dots, K\}$, we have $a_{\{i_1,i_2\}} = v_{i_1 \rightarrow \{i_2,j\}} + v_{i_2 \rightarrow \{i_1,j\}},$ $\forall j \in \{l\!+\!1, \dots, K\} \setminus \{i_1,i_2\}$, i.e., we have $(K\!-\!l\!-\!2)\binom{K-l}{2}$ equations in $(K\!-\!l\!-\!2)\binom{K-l}{2}$ unknowns. In turn, we have 
\begin{align}
 \sum_{j=l+1}^{K} \sum_{ \mc S \subset \{l\!+\!1, \dots, K\}\setminus\{j\}: |\mc S|=2} \! \! \! \! \! \! \! \! \! \! \! \! \! \! \! v_{j \rightarrow \mc S} = \! \bigg( \! \dfrac{K\!- \! l \!- \! 2}{2} \! \bigg)  \! \! \sum_{\mc S \subset \{l\!+\!1, \dots, K\}: |\mc S|=2} \! \! \! \! \! \! \! \! a_{\mc S} \! = \! \bigg(\! \dfrac{K\!- \! l \!- \! 2}{2} \! \bigg) \bigg(\sum_{i=l+1}^{K} m_i -1\bigg). 
\end{align}
\end{itemize}
Therefore, the delivery load due to multicast transmissions is given by
\begin{align}
 \sum_{j=l+1}^{K} \sum_{ \mc S \subset [K]\setminus\{j\}: |\mc S|=2} \! \! \! \! \! \! \! \! \! \! v_{j \rightarrow \mc S} &= \sum_{j=l+1}^{K} \left( \sum_{i \in [K]\setminus \{1,j\}} \! \! \! \! \! \! v_{j \rightarrow \{1,i\}} + \dots + \! \! \! \! \! \! \! \! \! \! \sum_{i \in [K]\setminus \{1,\dots,l,j\}} \! \! \! \! \! \! v_{j \rightarrow \{l,i\}} + \! \! \! \! \! \! \! \! \sum_{ \mc S \subset \{l\!+\!1, \dots, K\}\setminus\{j\}: |\mc S|=2} \! \! \! \! \! \! \! \! \! \! \! \!  \!  v_{j \rightarrow \mc S} \right), \\ &= \sum_{i=1}^{l} (K\!-\!i\!-\!1) m_i + \bigg(\dfrac{K\!- \! l \!- \! 2}{2}\bigg) \bigg(\sum_{i=l+1}^{K} m_i -1\bigg). \label{eqn_app_15}
\end{align}
We also need the following unicast transmissions.
\begin{itemize}
\item \underline{Unicast to user $1$}: 
\begin{align}
\! \! \! \! \! \sum_{j=l+1}^{K} v_{j \rightarrow \{1\}} &= \sum_{j=l+1}^{K} a_{\{j\}} + \sum_{i=2}^{l} \sum_{j=l+1}^{K} (a_{\{i,j\}} \! -\! a_{\{1,j\}}) + \! \! \! \! \! \! \sum_{\{i,j\}\subset \{l\!+\!1, \dots, K\}} \! \! \! \! \! \! \! \! \! \!  (a_{\{i,j\}} \!-\! a_{\{1,i\}} -\! a_{\{1,j\}}), \\ &= \bigg(2 \! - \! \sum_{k=1}^{K} m_k \bigg) \! +\!  \sum_{i=2}^{l} m_i \! + \! \bigg(\sum_{i=l+1}^{K} m_i \! - \! 1 \bigg) \! - \! (K\! -\!2)m_1 = 1 \!-\!(K\! -\!1)m_1.
\end{align}
\item \underline{Unicast to user $2$}: 
\begin{align}
\! \! \! \! \! \sum_{j=l+1}^{K} \! \! v_{j \rightarrow \{2\}} &= \! \! \sum_{j=l+1}^{K} \! \! a_{\{j\}} \! + \! \! \sum_{i=3}^{l} \sum_{j=l+1}^{K} \! \! (a_{\{i,j\}} \! -\! a_{\{2,j\}}) \! + \! \! \! \! \! \! \! \! \!  \sum_{\{i,j\}\subset \{l\!+\!1, \dots, K\}} \! \! \! \! \! \! \! \! \! \!  (a_{\{i,j\}} \!-\! a_{\{2,i\}} -\! a_{\{2,j\}}) \\ &=  1 \!-\!(K\! -\!2)m_2 \! - \! m_1.
\end{align}
\item \underline{Unicast to user $k \in \{3, \dots, l\}$}: Similarly, we have
\begin{align}
\! \! \! \! \! \sum_{j=l+1}^{K} v_{j \rightarrow \{l\}} &= 1 \! -\!(K\! -\!k)m_k \! - m_{k-1}-\dots -\! m_1.
\end{align}
\item \underline{Unicast to users $\{l\!+\!1, \dots, K\}$}:
\begin{align}
 \sum_{j=l+1}^{K} \sum_{i=l+1, i \neq j}^{K} v_{j \rightarrow \{i\}} = (K\!-\!l\!-\!1) \sum_{j=l+1}^{K}  a_{\{j\}} = (K\!-\!l\!-\!1) \bigg( 2 \! - \! \sum_{k=1}^{K} m_k\bigg).
\end{align}
\end{itemize}
Therefore, the delivery load due to unicast transmissions is given by
\begin{align}\label{eqn_app_16}
\sum_{j=l+1}^{K} \sum_{i=1, i \neq j}^{K} v_{j \rightarrow \{i\}} = l \! - \! \sum_{i=1}^{l} (K+l-2i) m_i + (K\!-\!l\!-\!1) \bigg( 2 \! - \! \sum_{k=1}^{K} m_k\bigg)
\end{align}
By adding (\ref{eqn_app_15}) and (\ref{eqn_app_16}), we get the total D2D delivery load given by (\ref{eqn_thm_mtot_2}).

\subsection{Achievability Proof of Theorem \ref{thm_mtot_K} }\label{proof_ach_thm_mtot_K}%
For $\sum_{i=1}^{K} m_i \geq K\!-\!1$, in the placement phase, each file $W_{n}$ is partitioned into subfiles $  \tilde W_{n,[K]\setminus \{i\}}, i \in [K]$ and $\tilde W_{n,[K]} $, such that 
\begin{subequations}
\begin{align}
&a_{[K]}=\sum_{i=1}^{K}m_i -(K\!-\!1), \ a_{[K]\setminus \{k\}}=1-m_k, \ \! k \in [K]. 
\end{align}
\end{subequations}
In the delivery phase, for $(K \! - \! l- \! 1)m_l < \sum_{i=l+1}^{K} m_i \! - \! 1$ and $(K \! - \! l- \! 2)m_{l+1} \geq \sum_{i=l+2}^{K} m_i \! - \! 1$, where $ l \in [K \!- \!2]$, we have the following transmissions 
\begin{align}
X_{K \rightarrow [i]}&= \oplus_{k \in [i]} W^{K \rightarrow [i]}_{d_k,[K]\setminus \{k\}}, \ i \in [l], \\
X_{j \rightarrow [K]\setminus \{j\}}&=\oplus_{k \in [K]\setminus \{j\}} W^{j \rightarrow [K]\setminus \{j\}}_{d_k,[K]\setminus \{k\}}, \ \! j \in \{l\!+\!1,\dots,K\}.
\end{align}
In particular, we have
\begin{align}
v_{K \rightarrow [i]} &= u_{[K]\setminus \{k\}}^{K \rightarrow [i]}= m_{i+1}-m_i, \ \! i \in [l\!-\!1], k \in [i], \\
v_{K \rightarrow [l]} &= u_{[K]\setminus \{k\}}^{K \rightarrow [l]}= \dfrac{1}{K \! - \! l \! - \! 1} \bigg( \sum_{j=l+1}^{K} m_j \! - \! 1 \! - (K\!-\!l\!-\!1)m_l \bigg), k \in [l], \\
v_{j \rightarrow [K]\setminus \{j\}} \! &= \! u_{[K]\setminus \{k\}}^{j \rightarrow [K]\setminus \{j\}} \! = \! \dfrac{(K\!-\!l\!-\!1)m_j \! + \! 1 \! - \! \! \sum\limits_{i=l+1}^{K} \! \! m_i }{K\! - \! l \! - \! 1}  , \ \! j \! \in \! \{l\!+\!1,\dots,K\}, k \! \in \! [K]\! \setminus \! \{j\}.
\end{align}
Therefore, the D2D delivery load is given by
\begin{align}
R^*_{\mathfrak{A},\mathfrak{D}}(\bm m)= v_{K \rightarrow [l]} + \sum_{i=1}^{l} v_{K \rightarrow [i]} + \sum_{j=l+1}^{K} v_{j \rightarrow [K]\setminus \{j\}} = 1-m_1.
\end{align}


\section{Optimality with Uncoded Placement}\label{sec_lowerbound}
In this section, we first prove the lower bound in Theorem \ref{thm_bound_genie}. Then, we   present the converse proofs for Theorems \ref{thm_equalcache}, \ref{thm_m1_lb}, and \ref{thm_mtot_2}.
\vspace{-0.1in}
\subsection{Proof of Theorem \ref{thm_bound_genie} }\label{proof_thm_genie}
\vspace{-0.05in}
First, we show that the D2D-based delivery assuming uncoded placement can be represented by $K$ index-coding problems, i.e., each D2D transmission stage is equivalent to an index-coding problem. In particular, for any allocation $\bm a \in \mathfrak{A}(\bm m) $, we assume that each subfile $\tilde W_{d_i,\mc S} $ consists of $|S|$ disjoint pieces $\tilde W_{d_i,\mc S}^{(j)}, \ \! j \in \mc S $, where $|\tilde W_{d_i,\mc S}^{(j)} |= a_{\mc S}^{(j)} F$ bits, i.e.,  $ a_{\mc S} = \sum_{j \in \mc S}  a_{\mc S}^{(j)} $. Additionally, the file pieces with superscript $(j)$ represent the messages in the $j$th index-coding problem. 
%
\begin{figure*}[t]
	\centering
	\begin{tabular}{cc}
	\hspace{-0.2in}
	 \subfloat[Index-coding setup.]{ \label{fig:indxcod_1}
				\includegraphics[scale=1.2]{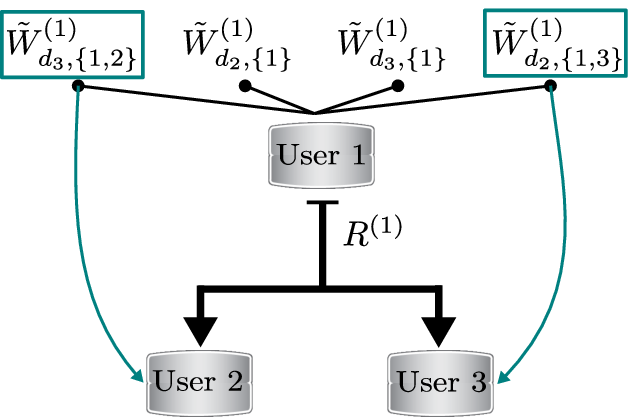} }
		
		& \hspace{-0.2in}
		\subfloat[Graph representation.]{ \label{fig:indxcod_1_grph}
			\includegraphics[scale=1.1]{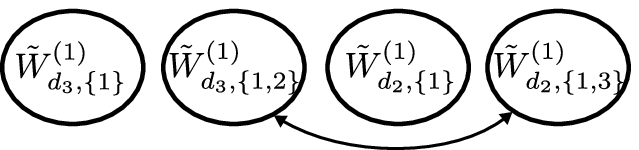}}
		\\
	\end{tabular}       
	\caption{Index-coding problem for $K=3$, and $j=1$.}
\vspace{-0.2in}
\end{figure*} 
%
For instance, consider the first index-coding problem in a three-user system, in which user $1$ acts as a server, see Fig. \subref*{fig:indxcod_1}. User $1$ needs to deliver $\tilde W_{d_2, \{1\}}^{(1)},\tilde W_{d_2, \{1,3\}}^{(1)}$ to user $2$, and  $\tilde W_{d_3,\{1\}}^{(1)}, \tilde W_{d_3, \{1,2\}}^{(1)} $ to user $3$. User $2$ has access to $\tilde W_{d_3, \{1,2\}}^{(1)} $, and user $3$ has access to $\tilde W_{d_2, \{1,3\}}^{(1)} $. The index coding problem depicted in Fig. \subref*{fig:indxcod_1} can be represented by the directed graph shown in Fig. \subref*{fig:indxcod_1_grph}, where the nodes represent the messages and a directed edge from $\tilde W_{*, \mc S}^{(1)} $ to $\tilde W_{d_i, *}^{(1)}$ exists if $i \in \mc S $ \cite{wan2016optimality}. Furthermore, by applying the acyclic index-coding bound \cite[Corollary 1]{arbabjolfaei2013capacity} on Fig. \subref*{fig:indxcod_1_grph}, we get 
\begin{align}
R^{(1)} F \geq  \sum_{i=1}^{K-1} \sum_{\mc S \subset [K] \ \! : \ \! 1 \in \mc S, \{ q_1, \dots, q_i\} \cap \mc S = \phi} |\tilde W_{d_{q_i}, \mc S}^{(1)} |,  \label{eqn_app_1}
\end{align}
where $\bm q  \in \mc P_{\{2,3\}}$ \cite{wan2016optimality,wan2017novel}.  In particular, for $K \! = \! 3$, we have 
\begin{align}
\! \! R^{(1)} F \geq |\tilde W_{d_2, \{1\}}^{(1)} | \! + \! |\tilde W_{d_3,\{1\}}^{(1)}| \! + \! |\tilde W_{d_2, \{1,3\}}^{(1)} |, \ \ \bm q=[2,3], \label{eqn_app_2}\\
\! \! R^{(1)} F \geq |\tilde W_{d_2, \{1\}}^{(1)} | \! + \! |\tilde W_{d_3,\{1\}}^{(1)}| \! + \! |\tilde W_{d_3, \{1,2\}}^{(1)}|, \ \ \bm q=[3,2]. \label{eqn_app_3}
\end{align}
Hence, for a given partitioning $a^{(j)}_{\mc S} $, by taking the convex combination of (\ref{eqn_app_2}), and (\ref{eqn_app_3}), we get   
\begin{align}
R^{(1)}( a^{(1)}_{\mc S},\alpha_{\bm q} ) \geq 2 a_{\{1\}}^{(1)} + \alpha_{[2,3]} a_{\{1,3\}}^{(1)}+ \alpha_{[3,2]} a_{\{1,2\}}^{(1)}, \label{eqn_app_4}
\end{align}
where $\alpha_{\bm q} \geq 0$, and $\alpha_{[2,3]}+\alpha_{[3,2]}=1 $. Similarly, we have
\begin{align}
R^{(2)}(a^{(2)}_{\mc S},\alpha_{\bm q} ) \geq 2 a_{\{2\}}^{(2)} + \alpha_{[1,3]} a_{\{2,3\}}^{(2)}+ \alpha_{[3,1]} a_{\{1,2\}}^{(2)}, \label{eqn_app_5} \\
R^{(3)}(a^{(3)}_{\mc S},\alpha_{\bm q} ) \geq 2 a_{\{3\}}^{(3)} + \alpha_{[1,2]} a_{\{2,3\}}^{(3)}+ \alpha_{[2,1]} a_{\{1,3\}}^{(3)}. \label{eqn_app_6}
\end{align}
Hence, for given $a^{(j)}_{\mc S} $ and $\alpha_{\bm q}$, the D2D delivery load $ \sum_{j=1}^{3} R^{(j)}(a^{(j)}_{\mc S},\alpha_{\bm q} )$ is lower bounded by the sum of the right-hand side of (\ref{eqn_app_4})-(\ref{eqn_app_6}). 
Furthermore, for $K$-user systems, $ R^{(j)}(a^{(j)}_{\mc S},\alpha_{\bm q} )$ is lower bounded by 
\begin{align}\label{eqn_app_12}
\tilde R^{(j)}( a^{(j)}_{\mc S},\alpha_{\bm q} ) \triangleq (K \! - \! 1) a^{(j)}_{\{j\}} \! + \! \! \! \! \sum_{\substack{ \mc S \subset [K] \ \! : \ \!  j \in \mc S, \\ 2 \leq |\mc S| \leq K -1}} \bigg( \sum_{i=1}^{K-|\mc S|} \! \!  \sum_{\substack{\bm q \in  \mc P_{[K]\setminus\{j\}} : \ \!  q_{i+1}  \in \mc S, \\ \{ q_1,\dots,q_{i} \} \cap \mc S=\phi } }  \! \! \! \! \! \! \! \! \! \! \! \! i \! \ \alpha_{\bm q}\bigg) a^{(j)}_{\mc S}. 
\end{align}
By taking the minimum over all feasible allocations and partitions, we get  
\begin{subequations} \label{eqn_app_7}
	\begin{align}
 R^*_{\mathfrak{A}}(\alpha_{\bm q}) \geq  \quad & \min_{a^{(j)}_{\mc S} \geq 0 }  
	& &  \! \! \sum_{j=1}^{K} \tilde R^{(j)}( a^{(j)}_{\mc S},\alpha_{\bm q} )\\
	& \text{subject to}
	& &   \! \! \! \! \! \! \sum\limits_{\mc S \subsetneq_{\phi}[K] } \sum_{j \in \mc S}    a^{(j)}_{\mc S}=1, \label{eqn_app_7_b} \\
	& & &  \! \! \! \! \! \! \! \! \! \! \! \! \sum\limits_{\mc S \subset [K]    :   k \in \mc S } \sum_{j \in \mc S} \! a^{(j)}_{\mc S} \! \leq \! m_k, \forall \ \! k \! \in \! [K]. \label{eqn_app_7_c}
	\end{align}
\end{subequations} 
The dual of the linear program in (\ref{eqn_app_7}) is given by
\begin{subequations}
	\begin{align}
& \max_{\lambda_{0} \in \mathbb{R},\lambda_{k} \geq 0}  
	& &  -\lambda_{0} - \sum_{k=1}^{K} m_k \lambda_{k} \\
	& \text{subject to}
	& &   \! \! \! \! \! \! \! \! \! \! \! \! \lambda_{0} + \sum_{k \in \mc S} \lambda_{k} +  \gamma_{\mc S} \geq 0,  \forall \ \!\mc S  \subsetneq_{\phi} [K], 
	\end{align}
\end{subequations}
where $\gamma_{\mc S}$ is defined in (\ref{eqn_gamma}), $\lambda_0$, and $\lambda_k$ are the dual variables associated with (\ref{eqn_app_7_b}), and (\ref{eqn_app_7_c}), respectively. Finally, by taking the maximum over all possible convex combinations $\alpha_{\bm q}, \forall \bm q \in \mc P_{[K]\setminus\{j\}}, \forall j \in [K] $, we get the lower bound in Theorem \ref{thm_bound_genie}.


\subsection{Converse Proof of Theorem \ref{thm_equalcache}}\label{proof_thm_equalcache}%
By substituting $\alpha_{\bm q} \! = \! 1/(K \! - \! 1) !$ in Theorem \ref{thm_bound_genie}, for $ 2 \leq |\mc S| \leq K\!- \!1$ we get 
\begin{align}
\gamma_{\mc S} &= \min_{j \in \mc S} \bigg\{ \sum\limits_{i=1}^{K-|\mc S|}   \sum\limits_{\substack{\bm q \in  \mc P_{[K]\setminus\{j\}}   : \ \!  q_{i+1}  \in \mc S, \\ \{ q_1,\dots,q_{i} \} \cap \mc S=\phi } }      i/(K \! - \! 1) ! \bigg\}, \\
 &= \sum\limits_{i=1}^{K-|\mc S|} \frac{i}{(K \! - \! 1) !} \binom{K-|\mc S|}{i} \ \! i ! \ \! (|\mc S|-1) \ \! (K \! - \! i - \! 2) ! = \frac{K \!- \! |\mc S|}{|\mc S|},
\end{align}
which follows from the number of vectors $\bm q \in  \mc P_{[K]\setminus\{j\}}$ such that $ q_{i+1}  \in \mc S, $ and $ \{ q_1,\dots,q_{i} \} \cap \mc S=\phi$. In particular, for given $j \! \in \! [K]$, $\mc S \! \subset \! [K] $ such that $j \in \mc S$, and $i \in \{1, \dots, K\! -\! |\mc S|\}$, there are $\binom{K-|\mc S|}{i} \ \! i ! $ choices for $\{q_1,\dots,q_{i} \}$,  $(|\mc S|-1) $ choices for $q_{i+1} $, and $(K \! - \! i - \! 2) ! $ choices for the remaining elements in $ [K] \setminus \big( \{ j\} \cup \{q_1,\dots,q_{i+1} \}\big)$. In turn, for $m_k=m, \ \! \forall k \in [K]$, we have 
\begin{subequations}\label{eqn_app_14} 
	\begin{align} 
 R^*_{\mathfrak{A}}(\bm m) \geq  \quad & \max_{\lambda_0 \in \mathbb{R},\lambda \geq 0 }  	& &  \! \! - \lambda_0 - K m \lambda\\
	& \text{subject to}
	& &   \! \!  \lambda_0 + l \lambda + (K-l)/l \geq 0, \forall \ \! l \! \in \! [K],
	\end{align}
\end{subequations}
which implies
\vspace{-0.15in} 
\begin{align} 
 R^*_{\mathfrak{A}}(\bm m) \geq  & \max_{\lambda \geq 0 } \left\{ \min_{ l \in [K]} \left\{(K-l)/l + \lambda \big( l-K m\big) \right\} \right\} , 	
\end{align}
In particular, for $m= t/K$ and $t \in [K]$, we have  
\vspace{-0.1in}
\begin{align}
R^*_{\mathfrak{A}}(\bm m) \geq \max_{\lambda \geq 0 } \Big\{ \min  \big\{(K \! - \! 1) - (t\! - \! 1)\lambda, \dots, (K \! - \! t)/t, \dots,\lambda K(1\! - \!m) \big\} \Big\}= (K-t)/t,
\end{align}
since this piecewise linear function is maximized by choosing $ \frac{K}{t(t+1)} \leq \lambda^* \leq \frac{K}{t(t-1)}$. In general, for $m= (t+\theta)/K$ and $ 0 \leq \theta \leq 1 $, we get 
\vspace{-0.1in}
\begin{align} 
 R^*_{\mathfrak{A}}(\bm m) &\geq \max_{\lambda \geq 0 } \Big\{ \min  \big\{\dots,(K \! -\! t)/t - \theta \lambda, \ (K\! -\! t \!- \!1)/(t \! + \! 1) - (1-\theta) \lambda, \dots\big\} \Big\} \\
 &= \dfrac{K-t}{t}-\dfrac{\theta K}{t(t+1)}=\dfrac{K-t}{t}-\dfrac{(K m -t) K}{t(t+1)}, 	
\end{align}
which is equal to (\ref{eqn_thm_equalcache}).

\subsection{Converse Proof of Theorem \ref{thm_m1_lb}}\label{proof_conv_thm_m1_lb}%

Similarly, by substituting $\alpha_{\bm q} \! = \! 1/(K \! - \! 1) !$ in Theorem \ref{thm_bound_genie}, we get 
\begin{subequations}\label{eqn_app_14} 
	\begin{align} 
 R^*_{\mathfrak{A}}(\bm m) \geq  \quad & \max_{\lambda_0 \in \mathbb{R},\lambda_j \geq 0 }  	& &  \! \! - \lambda_0 - \sum_{j=1}^{K} \lambda_j m_j \\
	& \text{subject to}
	& &   \! \!  \lambda_0 + \sum_{i \in \mc S} \lambda_i + (K-|\mc S|)/|\mc S| \geq 0. \forall \ \! \mc S \subsetneq_{\phi}[K],
	\end{align}
\end{subequations}
In turn, for $ t \leq \sum_{j=1}^{K} m_j \leq t\!+\!1$, we get
\vspace{-0.1in}
\begin{align} 
 R^*_{\mathfrak{A}}(\bm m) \! &\geq \!  \max_{\lambda \geq 0 } \left\{ \min_{ l \in [K]} \left\{ \! (K\!-\!l)/l \! + \! \lambda \big( l \! - \! \! \sum_{j=1}^{K} m_j \big) \right\} \right\}  \! = \! \dfrac{t K \! + \! (t \! + \! 1)(K \! -\! t)}{t(t\!+\!1)} - \dfrac{K \sum_{j=1}^{K} m_j}{t(t\!+\!1)}.
\end{align}

\subsection{Converse Proof of Theorem \ref{thm_mtot_2}}\label{proof_conv_thm_mtot_2}%
By substituting, $\alpha_{\bm q} \! = \! 1$ for $j \in [l] $, $\bm q=[1,2, \dots,j\!-\!1,j\!+\!1, \dots,K] $, and $\alpha_{\bm q} \! = \! 1/(K\!-\!l\!-\!1)!$ for $j \in \{l\!+\!1,\dots,K\} $, $\bm q=[1,\dots,l,\bm x] $, $ \forall \bm x \in \mc P_{\{l+1,\dots,K\}\setminus \{j\}}$, in Theorem \ref{thm_bound_genie}, we get 
\begin{align}
\gamma_{\mc S} \triangleq \begin{cases} K-1, \text{ for } |\mc S| =1,\\
\dfrac{K+l(|S|\!-\!1)+|S|}{|S|}, \text{ for } \mc S \subset \{l\!+\!1,\dots,K\} \text{ and } 2 \leq \! |\mc S| \! \leq K \! - \! 1, \\
\min\limits_{i \in \mc S} i -1, \text{ for }  \mc S \cap [l] \neq \phi \text{ and } 2 \leq \! |\mc S| \! \leq K \! - \! 1, \\
0, \text{ for } \mc S =[K].
\end{cases}
\end{align}
In turn, $\lambda_0 = - (3K-l-2)/2$, $\lambda_j = K  -  j $ for $j \in [l]$, and $\lambda_j = (K-l)/2 $ for $j \in \{l+1,\dots,K\}$, is a feasible solution to (\ref{eqn_bound_genie}).
\begin{remark} In this region, we achieve the tightest lower bound by choosing $\alpha_{\bm q}$, taking into consideration that the delivery load depends on the individual cache sizes of the users in $[l]$ and the aggregate cache size of the users in $\{l\!+\!1,\dots,K\}$.
\end{remark}

%
%
%
%
%
%

\begin{figure}[t]
\vspace{-.05 in}
\includegraphics[scale=0.65]{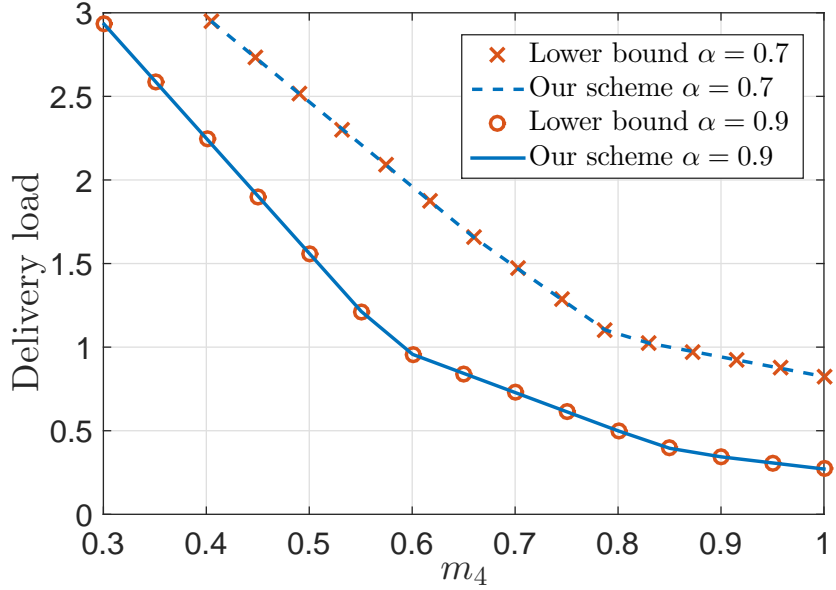}
\centering
\vspace{-.1 in}
\caption{Comparing $R^*_{\mathfrak{A},\mathfrak{D}}(\bm m)$, and lower bound on $R^*_{\mathfrak{A}}(\bm m) $, for $K=N=4$, and $ m_k=\alpha \ \! m_{k+1}$.}\label{fig:comp_dlv}
\end{figure}

\vspace{-0.1in}
\section{Discussion}\label{sec_disc}

\subsection{The D2D Delivery Load Memory Trade-off}

In Section \ref{sec_results}, we have characterized the D2D delivery load memory trade-off with uncoded placement, $R^*_{\mathfrak{A}}(\bm m)$, for several special cases. 
For general systems, we observe numerically that the proposed caching scheme coincides with the delivery load uncoded placement lower bound given in Theorem \ref{thm_bound_genie}. For example, in Fig. \ref{fig:comp_dlv}, we compare the D2D delivery load $R^*_{\mathfrak{A},\mathfrak{D}}(\bm m) $ achievable with our proposed caching scheme with the lower bound on $R^*_{\mathfrak{A}}(\bm m)$ in Theorem \ref{thm_bound_genie}, for $K \! = \! N \! = \! 4$  and $m_k=\alpha \ \! m_{k+1}$, and observe they coincide. 

\subsection{Comparison between Server-based and D2D-based Delivery Loads}
By comparing the server-based system \cite{maddah2014fundamental,ibrahim2018coded} delivery load and D2D-based system delivery load, we observe the following:
\begin{itemize}
\item The D2D-based delivery load memory trade-off with uncoded placement, $R^*_{\mathfrak{A},\text{D2D}}(K,\frac{m_{\text{tot}}}{K}) $, for a system with $K$ users and equal cache size $m=m_{\text{tot}}/K$, is equal to the server-based delivery load memory trade-off assuming uncoded placement for a system with $K\!-\!1$ users and cache size $m=(m_{\text{tot}}\!-\!1)/(K\!-\!1)$, which we denote by $R^*_{\mathfrak{A},\text{Ser}}(K\!-\!1, \frac{m_{\text{tot}}-1}{K-1}) $ \cite{maddah2014fundamental}. In particular, for $m_{\text{tot}} \in [K]$, we have 
\begin{align}
R^*_{\mathfrak{A},\text{Ser}}\Big(K\!-\!1, \frac{m_{\text{tot}}\!-\!1}{K\!-\!1}\Big)=\dfrac{(K\!-\!1)(1-\frac{m_{\text{tot}}-1}{K-1})}{1+(K\!-\!1)(\frac{m_{\text{tot}}-1}{K-1})}=\dfrac{1-\frac{m_{\text{tot}}}{K}}{\frac{m_{\text{tot}}}{K}}=R^*_{\mathfrak{A},\text{D2D}}\Big(K,\frac{m_{\text{tot}}}{K}\Big). 
\end{align}
\item From Theorem \ref{thm_m1_lb}, we conclude that if $m_{\text{tot}} \triangleq \sum\limits_{k=1}^{K} m_k$ and $m_1 \geq (m_{\text{tot}}\!-\!1)/(K\!-\!1)$, then the D2D delivery load memory trade-off with uncoded placement, $R^*_{\mathfrak{A},\text{D2D}}(K,\bm m) $, for a system with $K$ users and distinct cache sizes $\bm m$, is equal to $R^*_{\mathfrak{A},\text{D2D}}(K,\frac{m_{\text{tot}}}{K} ) $. In turn, if $m_1 \geq (m_{\text{tot}}-1)/(K\!-\!1)$, then $R^*_{\mathfrak{A},\text{D2D}}(K,\bm m)= R^*_{\mathfrak{A},\text{Ser}}(K\!-\!1, \frac{m_{\text{tot}}-1}{K-1}) $.
\item For a $K$-user D2D system with $m_K\!=\!1$, the D2D delivery load $R^*_{\mathfrak{A},\text{D2D}}(K,[m_1,\dots,m_{K-1},1])$ is equal to $R^*_{\mathfrak{A},\text{Ser}}(K\!-\!1, [m_1,\dots,m_{K-1}])$. For example, for $K\!=\!3$, we have
\begin{align}
R^*_{\mathfrak{A},\text{D2D}}(3,[m_1,m_{2},1])=R^*_{\mathfrak{A},\text{Ser}}(2, [m_1,m_{2}])=\max \left\{ 2 \! - \! 2 m_1 \! - \! m_2, \ 1  - m_1\right\}.
\end{align}
\end{itemize}


\subsection{Connection between Coded Distributed Computing and D2D Coded Caching Systems}
In coded distributed computing (CDC) systems, the computation of a function over the distributed computing nodes is executed in two stages, named \textit{Map} and \textit{Reduce} \cite{li2018fundamental}. In the former, each computing node maps its local inputs to a set of intermediate values. In order to deliver 
the intermediate values required for computing the final output at each node, the nodes create multicast transmissions by exploiting the redundancy in computations at the nodes. In the latter, each node reduces the intermediate values retrieved from the multicast signals and the local intermediate values to the desired final outputs. 

For CDC systems where the nodes are required to compute different final outputs, the CDC problem can be mapped to a D2D coded caching problem, where the cache placement scheme is uncoded and symmetric over the files \cite{li2018fundamental,wan2018fundamental}. Therefore, the D2D caching scheme proposed in this work can be utilized in heterogeneous CDC systems where the nodes have varying computational/storage capabilities \cite{kiamari2017heterogeneous}. The mapping between the two problems is described in the following remark.
\begin{remark}\label{remark_CDC} A D2D caching system with $K$ users, $N$ files, each with size $F$ symbols, where $m_k$ is the normalized cache size at user $k$, corresponds to a CDC system with $K$ nodes, $F$ files, $N$ final outputs, where $\tilde M_k =m_k F$ is the number of files stored at node $k$. More specifically, in the map stage, node $k$ computes $N$ intermediate values for each cached file. In the reduce stage, node $k$ computes $N/K$ final outputs from the local intermediate values combined with those retrieved from the multicast signals.
\end{remark}
\begin{remark} Reference \cite{kiamari2017heterogeneous} derived the optimal communication load in a heterogeneous CDC system consisting of three nodes with different computational/storage capabilities. As a consequence of Remark \ref{remark_CDC}, the optimal communication load found in \cite{kiamari2017heterogeneous} is the same as the minimum worst-case D2D delivery load with uncoded placement in Theorem \ref{thm_3ue}.
\end{remark}

\vspace{-0.05in}

\section{Conclusions}\label{sec_concl}
In this paper, we have proposed a coded caching scheme that minimizes the worst-case delivery load for D2D-based content delivery to users with unequal cache sizes. We have derived a lower bound on the delivery load with uncoded placement. We have proved the optimality of our delivery scheme for several cases of interest. In particular, we explicitly characterize $R^*_{\mathfrak{A}}(\bm m) $ for the following cases: \begin{enumerate*}[label=(\roman*)]
  \item $m_k \! = \! m, \forall k$,
  \item $(K\!-\!2) m_1 \! \geq \! \sum_{k=2}^{K} m_k \! - \! 1$,
  \item $ \sum_{k=1}^{K} m_k \! \leq \! 2$,
  \item $  \sum_{k=1}^{K} m_k \! \geq \! K\!-\!1$, and 
  \item $K\!=\!3$.
\end{enumerate*}
More specifically, for $m_k \! = \! m, \forall k$, we have shown the optimality of the caching scheme in \cite{ji2016fundamental}. We have also shown that the minimum delivery load depends on the sum of the cache sizes and not the individual cache sizes if the smallest cache size satisfies $(K\!-\!2) m_1 \! \geq \! \sum_{k=2}^{K} m_k \! - \! 1$. 

In the small total memory regime where $ \sum_{k=1}^{K} m_k \! \leq \! 2$, we have shown that there exist $K-1$ levels of heterogeneity and in the $l$th heterogeneity level $R^*_{\mathfrak{A}}(\bm m) $ depends on the individual cache sizes of users $\{1,\dots,l\}$ and the sum of the cache sizes of remaining users. In the large total memory regime where $  \sum_{k=1}^{K} m_k \! \geq \! K\!-\!1$ and $(K\!-\!2) m_1 \! < \! \sum_{k=2}^{K} m_k \! - \! 1$, we have shown that our caching scheme achieves the minimum delivery load assuming general placement. That is, it coincides with the cut-set bound \cite{ji2016fundamental}. We have articulated the relationship between the server-based and D2D delivery problems. Finally, we have discussed the coded distributed computing (CDC) problem \cite{li2018fundamental} and how our proposed D2D caching scheme can be tailored for heterogeneous CDC systems where the nodes have unequal storage.

Future directions include considering heterogeneity in cache sizes and node capabilities for hierarchical cache-enabled networks and general network topologies.



\vspace{-0.1in}
\appendices

\section{Achievability Proof of Theorem \ref{thm_3ue} }\label{app_thm_3ue_ach}
\subsection*{\underline{Region \rm{I}:} $ 1 \leq m_1+m_2+m_3 \leq 2$ and $ m_1 \geq m_2+m_3-1 $}
In this region, we show that there exists a feasible solution to (\ref{eqn_opt}) that achieves $ R^*_{\mathfrak{A},\mathfrak{D}}(\bm m)= \frac{7}{2}  - \frac{3}{2} \big( m_1+m_2+m_3\big)$. In particular, we consider the caching schemes described by $v_{1 \rightarrow \{2\}} = v_{1 \rightarrow \{3\}} = a_{\{1\}}$, $v_{2 \rightarrow \{1\}} = v_{2 \rightarrow \{3\}} = a_{\{2\}}$, $v_{3 \rightarrow \{1\}} = v_{3 \rightarrow \{2\}} = a_{\{3\}}$, $v_{1 \rightarrow \{2,3\}} + v_{2 \rightarrow \{1,3\}} = a_{\{1,2\}}$, $v_{1 \rightarrow \{2,3\}} + v_{3 \rightarrow \{1,2\}} = a_{\{1,3\}}$, $v_{2 \rightarrow \{1,3\}} + v_{3 \rightarrow \{1,2\}} = a_{\{2,3\}}$, and $a_{\{1,2,3\}}=0 $. In turn, the placement feasibility conditions in (\ref{eqn_feas_alloc}) reduce to
\begin{subequations}\label{eqn_app_22} 
\begin{align}
 v_{1 \rightarrow \{2,3\}} + v_{2 \rightarrow \{1,3\}}+ v_{3 \rightarrow \{1,2\}}&=\big( m_1+m_2+m_3-1\big)/2, \\
a_{\{1\}}+ v_{1 \rightarrow \{2,3\}}&=\big( m_1+1-m_2-m_3\big)/2, \\
a_{\{2\}}+ v_{2 \rightarrow \{1,3\}}&=\big( m_2+1-m_1-m_3\big)/2, \\
a_{\{3\}}+ v_{3 \rightarrow \{1,2\}}&=\big( m_3+1-m_1-m_2\big)/2.
\end{align}
\end{subequations}
Note that any caching scheme satisfying (\ref{eqn_app_22}), achieves the D2D delivery load 
\begin{align}
R^*_{\mathfrak{A},\mathfrak{D}}(\bm m) \! &= \! 2 \big( a_{\{1\}} \! + \! a_{\{2\}} \! + \!a_{\{3\}}\big) \! + \! v_{1 \rightarrow \{2,3\}} \! + \! v_{2 \rightarrow \{1,3\}} \! + \! v_{3 \rightarrow \{1,2\}}= \frac{7}{2}  - \frac{3}{2} \big( m_1+m_2+m_3\big)
\end{align}
In turn, we only need to choose a non-negative solution to (\ref{eqn_app_22}), for instance we can choose $a_{\{j\}}=\rho_j \big( 2-m_1-m_2-m_3\big), $ 
such that $\sum_{j=1}^{3} \rho_j =1$, and $ 0 \leq \rho_j \leq \dfrac{2 m_j +1 - \sum_{i=1}^{3} m_i}{2 \big( 2-\sum_{i=1}^{3} m_i\big)}$.

\vspace{-0.1in}
\subsection*{\underline{Region \rm{II}:} $ 1 \leq m_1+m_2+m_3 \leq 2$ and $m_1 < m_2+m_3-1 $}
In this region, we achieve the D2D delivery load $ R^*_{\mathfrak{A},\mathfrak{D}}(\bm m)= 3  - 2 m_1-m_2-m_3$, by considering the caching schemes described by $v_{1 \rightarrow \{2\}} = v_{1 \rightarrow \{3\}} = a_{\{1\}}=0$, $v_{2 \rightarrow \{1\}} = v_{2 \rightarrow \{3\}} = a_{\{2\}}$, $ v_{3 \rightarrow \{2\}} = a_{\{3\}}$, $v_{3 \rightarrow \{1\}} = a_{\{3\}}+ \big( a_{\{2,3\}}-a_{\{1,2\}}-a_{\{1,3\}}\big)$, $v_{1 \rightarrow \{2,3\}}  = 0$, $v_{2 \rightarrow \{1,3\}} = a_{\{1,2\}}$, $v_{3 \rightarrow \{1,2\}} = a_{\{1,3\}}$, $a_{\{2,3\}}=m_2+m_3-1 $ and $a_{\{1,2,3\}}=0 $. Hence, we only need to choose a non-negative solution to the following equations
\begin{align}\label{eqn_app_23} 
a_{\{2\}}+a_{\{1,2\}}=1-m_3, \ a_{\{3\}}+a_{\{1,3\}}&=1- m_2, \
a_{\{1,2\}}+a_{\{1,3\}}=m_1,
\end{align}
which follows from (\ref{eqn_feas_alloc}). Note that any non-negative solution to (\ref{eqn_app_23}), achieves $R^*_{\mathfrak{A},\mathfrak{D}}(\bm m)= 3  - 2 m_1-m_2-m_3 $. For instance, we can choose $a_{\{1,3\}}=0$ when $m_1 +m_3 \leq 1$ and $a_{\{2\}}=0$ when $m_1 +m_3 > 1$.
\vspace{-0.15in}
\subsection*{\underline{Region \rm{III}:} $ m_1+m_2+m_3 > 2$ and $m_2+m_3 \leq 1+m_1 $}
In order to achieve $ R^*_{\mathfrak{A},\mathfrak{D}}(\bm m)= \frac{3}{2}  - \frac{1}{2} \big( m_1+m_2+m_3\big)$, we consider the caching scheme described by $v_{1 \rightarrow \{2,3\}} \! = \! (m_1 \! + \! 1 \! - \! m_2 \! - \! m_3)/2$, $v_{2 \rightarrow \{1,3\}} \! = \! (m_2 \! + \! 1 \! - \! m_1 \! - \! m_3)/2$, $v_{3 \rightarrow \{1,2\}} \! = \! (m_3 \! + \! 1 \! - \! m_1 \! - \! m_2)/2$, $a_{\{1,2\}} \! = \! 1 \! - \! m_3$, $a_{\{1,3\}} \! = \! 1 \! - \! m_2,$ $a_{\{2,3\}} \! = \! 1 \! - \! m_1$, and $a_{\{1,2,3\}} \! = \! m_1 \! + \! m_2 \! + \! m_3 \! - \! 2$.  
%
%
\vspace{-0.3in}
\subsection*{\underline{Region \rm{IV}:} $ m_1+m_2+m_3 > 2$ and $m_2+m_3 > 1+m_1 $}
Finally, $ R^*_{\mathfrak{A},\mathfrak{D}}(\bm m)= 1-m_1$ is achieved by $a_{\{1,2\}} \! = \! 1 \! - \! m_3$, $a_{\{1,3\}} \! = \! 1 \! - \! m_2$, $a_{\{2,3\}} \! = \! 1 \! - \! m_1$, 
$a_{\{1,2,3\}} \! = \! m_1 \! + \! m_2 \! + \! m_3 \! - \! 2$, $ v_{3 \rightarrow \{1\}} \! = \! m_2 +m_3 - m_1 -1$, $v_{2 \rightarrow \{1,3\}} \! = \! 1-m_3$, and $v_{3 \rightarrow \{1,2\}} \! = \! 1-m_2$.  

\vspace{-0.1in}
\section{Converse Proof of Theorem \ref{thm_3ue}  }\label{app_thm_3ue_conv}
By substituting $\alpha_{\bm q}=1/2, \forall \bm q \in \mc P_{[3]\setminus\{j\}}, \forall j \in [3]  $ in Theorem \ref{thm_bound_genie}, we get
\begin{subequations}\label{eqn_app_9} 
	\begin{align} 
 R^*_{\mathfrak{A}}(\bm m) \geq  \quad & \max_{\lambda_{0} \in \mathbb{R},\lambda_{k} \geq 0 }  	& &  \! \! - \lambda_0 - \lambda_1 m_1 - \lambda_2 m_2- \lambda_3 m_3\\
	& \text{subject to}
	& &   \! \!  \lambda_0 +  \lambda_j +2 \geq 0, \forall \ \! j \! \in \! [3],\\
	& & &  \! \! \! \!  \lambda_0 +  \lambda_i + \lambda_j +1/2 \geq 0, \forall \ \! i \! \in \! [3], j \neq i,\\
	& & &  \! \! \! \!  \lambda_0 +  \lambda_1 + \lambda_2 + \lambda_3 \geq 0. 
	\end{align}
\end{subequations}
By choosing two feasible solutions to (\ref{eqn_app_9}), we get
\begin{align}
  R^*_{\mathfrak{A}}(\bm m) \geq \max \left\{ \dfrac{7}{2}  - \frac{3}{2} \big( m_1 \!+ \! m_2 \! + \! m_3\big), \ \dfrac{3}{2}  - \frac{1}{2} \big( m_1 \! + \! m_2 \! + \! m_3\big)\right\}.
\end{align}
Similarly, by substituting $\alpha_{[2,3]}=\alpha_{[1,3]}=\alpha_{[1,2]}=1 $ in Theorem \ref{thm_bound_genie}, we can show that
\begin{align}
  R^*_{\mathfrak{A}}(\bm m) \geq \max \left\{ 3 \! - \! 2 m_1 \! - \! m_2 \! - \! m_3, \ 1  - m_1\right\}.
\end{align}

\bibliographystyle{IEEEtran}
\bibliography{IEEEabrv,references}
 
\end{document}